\documentclass[aps, pra, a4paper,showpacs,twocolumn,10pt]{revtex4-1}
\usepackage{bbm, amsmath, amssymb, amsthm, bm,textcomp, nicefrac,geometry,ragged2e}%, subcaption}
\usepackage{graphicx,epstopdf,color,verbatim,enumitem}
\usepackage[dvipsnames]{xcolor}

\usepackage[bbgreekl]{mathbbol}
\usepackage{graphicx,epstopdf,color,verbatim,enumitem,ulem}

\definecolor{myurlcolor}{rgb}{0,0,0.7}
\definecolor{myrefcolor}{rgb}{0.8,0,0}
\usepackage[unicode=true,pdfusetitle, bookmarks=false,bookmarksnumbered=false,
bookmarksopen=false, breaklinks=false,pdfborder={0 0 0},backref=false,
colorlinks=true, linkcolor=myrefcolor,citecolor=myurlcolor,urlcolor=myurlcolor]
{hyperref}

\usepackage{pbox,hyperref,array}
\usepackage[caption=false]{subfig}

\newcommand{\ket}[1]{\left|#1\right\rangle}

\def\E{\mathcal{E}}
\newcommand{\be}{\begin{equation}}
\newcommand{\ee}{\end{equation}}
\definecolor{darkorange}{RGB}{255,140,0}

\newcommand{\bea}{\begin{eqnarray}}
\newcommand{\eea}{\end{eqnarray}}

\def\tr{\mathrm{tr}}

\makeatletter
\newtheorem*{rep@theorem}{\rep@title}
\newcommand{\newreptheorem}[2]{%
\newenvironment{rep#1}[1]{%
 \def\rep@title{#2 \ref{##1}}%
 \begin{rep@theorem}}%
 {\end{rep@theorem}}}
\makeatother

\newreptheorem{theorem}{Theorem}
\newtheorem{lemma}{Lemma}

\newtheorem*{result*}{Result}

\newcommand{\F}{\mathcal{F}}
\renewcommand{\P}{\mathrm{Pr}}

\newcommand{\px}{\sigma_x}

\newcommand{\pz}{\sigma_z}

\newcommand{\ptr}[2]{\text{tr}_{#1}\left[ #2 \right]}

\newcommand{\id}{id}

\newcommand{\NN}{\mathcal{N}}
\def\ket#1{\left| #1 \right\rangle}

\def\dm#1{\left|#1 \right\rangle \left\langle #1 \right|}
\newcommand{\ketbra}[3]{\left| #1 \right\rangle \left\langle #2 \right|_{\text{#3}}}

\newtheorem{prop}{Proposition}\def\PRO{\begin{prop}}\def\ORP{\end{prop}}
\newtheorem{coro}{Corollary}\def\COR{\begin{coro}}\def\ROC{\end{coro}}
\newtheorem{theo}{Theorem}\def \TH{\begin{theo}}\def\HT{\end{theo}}
\def\TH{\begin{theo}}\def\HT{\end{theo}}
\newtheorem{defi}[prop]{Definition}\def\DE{\begin{defi}}\def\ED{\end{defi}}
\newtheorem{lemme}[prop]{Lemma}\def\LE{\begin{lemme}}\def\EL{\end{lemme}}

\DeclareGraphicsExtensions{.png,.pdf,.eps}

\begin{document}
%\title{Private entanglement via hashing}
%\title{Simple proof of confidentiality of privacy distillation}
\title{Simple proof of confidentiality for private quantum channels in noisy environments}
\author{A. Pirker$^{1}$}\thanks{These authors contributed equally}
\author{M. Zwerger$^{1,2}$}\thanks{These authors contributed equally}
\author{V. Dunjko$^{1,3}$}
\author{H. J. Briegel$^{1,4}$}
\author{W.~D\"ur$^1$}
\affiliation{$^1$ Institut f\"ur Theoretische Physik, Universit\"at Innsbruck, Technikerstra{\ss}e 21a, A-6020 Innsbruck, Austria\\
$^2$ Departement Physik, Universit\"at Basel, Klingelbergstra{\ss}e 82, 4056 Basel, Switzerland\\
$^3$ Max-Planck-Institute of Quantum Optics, Hans-Kopfermann-Strasse 1, D-85748 Garching, Germany \\
$^4$ Fachbereich Philosophie, Universit\"at Konstanz, Universit\"atsstra\ss e 10, 78464 Konstanz,$\,$Germany}

\date{\today}

\begin{abstract}
Complete security proofs for quantum communication protocols can be notoriously involved, which convolutes their verification, and obfuscates the key physical insights the security finally relies on. In such cases, for the majority of the community, the utility of such proofs may be restricted. Here we provide a simple proof of confidentiality for parallel quantum channels established via entanglement distillation based on hashing, in the presence of noise, and a malicious eavesdropper who is restricted only by the laws of quantum mechanics. The direct contribution lies in improving the linear confidentiality levels of recurrence-type entanglement distillation protocols to exponential levels for hashing protocols. The proof directly exploits the security relevant physical properties: measurement-based quantum computation with resource states and the separation of Bell-pairs from an eavesdropper. The proof also holds for situations where Eve has full control over the input states, and obtains all information about the operations and noise applied by the parties. The resulting state after hashing is private, i.e., disentangled from the eavesdropper. Moreover, the noise regimes for entanglement distillation and confidentiality do not coincide: Confidentiality can be guaranteed even in situation where entanglement distillation fails. We extend our results to multiparty situations which are of special interest for secure quantum networks.  
\end{abstract}
\maketitle

\section{Introduction}

Secure and private quantum communication is a concept of fundamental importance for emerging quantum technologies. The secure generation of a secret key for the encryption of classical data has received enormous attention in recent years \cite{ShorQKD,RennerPhd,ZhaoQKD,GottesmanQKD,LoQKD,LongGottesman,LongLo}, and is believed to be one of the key applications of quantum information science. Security has been shown under ever more general assumptions, finally arriving at device-independent proofs where the devices for secret key expansion are not trustworthy \cite{AcinDI,LimDI,UmeshDI}.  However, while establishing entanglement between two remote parties served as key ingredient in many security proofs of QKD, most existing proofs are not established by sharpening this intuition, i.e. they follow a more convoluted, tedious, and less straightforward route  \cite{RennerPhd,LongLeverrier,LongBiham,LongMayers}. 
\newline
Here we consider the problem of confidential or secure transmission of quantum information via quantum channels, equally important as QKD but far less studied. This task is closely related to the confidential generation of maximally entangled, distributed quantum states. Both are essential ingredients of quantum networks \cite{AcinNetwork,MeterNetwork,KimbleQInet}, quantum key agreement protocols \cite{XuQka,SunQka1,SunQka2}, and distributed quantum computation \cite{CiracDistributed}. In an idealized, noiseless situation a secure quantum channel, studied in \cite{SecChannelPortmann,SecChannelGarg,SecChannelBroadbent}, may be established in terms of teleporation \cite{BennettTele} using a perfect Bell-pair. The situation turns out to be far less straightforward in a noisy scenario. Nevertheless, it was shown that private entanglement is feasible when considering noisy channels and perfect operations \cite{Deutsch,BennettRecurrence}, as well as noise in local operations for independent and identically distributed (i.i.d.) \cite{Aschauer02} and non i.i.d. \cite{Pirker16a} situations. The latter works consider the recurrence-type entanglement distillation protocols  \cite{Deutsch,BennettRecurrence}, which probabilistically increase the fidelity and factor out any eavesdropper with a linear rate of convergence in terms of initial states. \newline 
Hashing protocols \cite{Be96,Du05,EPPallGraphs,ChenEPPGraph,ManevaEPP,HostensEPPGraph,GlancyEPPGraph,HostensEPPGraph2} are one-way entanglement distillation protocols (EDP) which overcome these limitations. They are deterministic and converge {\it exponentially fast} in terms of initial states towards several copies of a maximally entangled state. This enables for several confidential quantum channels in parallel, crucial for big quantum data transmission \cite{Zw17} and which is in contrast to recurrence-type entanglement distillation protocols. \newline
In this paper we provide a proof of confidentiality for hashing protocols in a noisy setting where the eavesdropper has full control over all the initial states. Since the confidentiality of recurrence-type entanglement distillation protocols \cite{Deutsch,BennettRecurrence} has been shown in similar scenarios \cite{Pirker16a}, this alone is not too surprising, even though hashing enables for exponential confidentiality levels rather than linear ones. Nevertheless, due to the simplicity of the confidentiality proof we clearly identify the relevant elements of physical properties from which the formal claim follows: the purity of the target state for noiseless distillation protocols and the way one deals with noise in measurement-based quantum computation (MBQC) with resource states. We emphasize that both are not exploitable in a noisy gate-based implementation as we illustrate later. The interest of using such characteristics, arguably, goes beyond the direct cryptographic statement they are implying. What is more, we identify a regime of noise where privacy, or equivalently confidentiality, is feasible, whereas distillation is not. Furthermore we show that hashing establishes privacy even when the eavesdropper is provided with information regarding all noise processes occurring in Alice's and Bob's laboratory, which is one step towards device independence for protocols with a quantum output.\newline
Early security proofs for QKD \cite{LongLo} rely on fault-tolerant quantum computation to reduce the problem of proving security to a noiseless setting, and utilize quantum random hashing \cite{Be96} to verify the successful generation of entanglement. In contrast, our approach eliminates the necessity of fault-tolerant quantum computation by exploiting physical properties of MBQC, and we use hashing as an active tool to establish high-fidelity entangled pairs via entanglement distillation rather than verifying them. Other works \cite{GottesmanQKD,ShorQKD,LongGottesman} also use the existence of (one-way) entanglement distillation protocols. However, earlier works \cite{LongLo,ShorQKD,LongGottesman} lack a full treatment of the finite size setting, crucial for realistic regimes \cite{LongLeverrier}. In contrast, here we analyse the finite size performance of hashing and explicitly provide confidentiality levels also in non-i.i.d. scenarios.  \newline
Entanglement distillation protocols aim at distilling entanglement from a noisy ensemble of bi- or multipartite quantum states via local operations and measurements. Hashing protocols \cite{Be96,Du05,EPPallGraphs,ChenEPPGraph,ManevaEPP,HostensEPPGraph,GlancyEPPGraph,HostensEPPGraph2} form a specific subset of those protocols, which rely on the concept of likely subspaces \cite{Schumacher}, used in information theory, and universal hash functions \cite{UniversalHash}, typically applied in the context of privacy amplification. Their operation is usually described on a large, noisy ensemble (called initial states) and one distills in the asymptotic limit a fraction of systems in a maximally entangled state, see Appendix \ref{app:bipartite:hashing} for more details. However, it was shown that hashing via quantum gates fails in the presence of noise \cite{Zw14H}. This drawback is overcome by measurement-based quantum information processing \cite{RausMeasQCcluster}. There, the desired quantum operation is realized via Bell-measurements between the input quantum state and the input qubits of a resource state, referred to as read-in measurements. Consequently the only source of noise within this computational approach is due to imperfect resource states and noisy Bell-measurements (which can be accounted for by an increased level of the noise acting on the resource state, see \cite{Zw14H}). A measurement-based implementation of the hashing protocol, see Appendix \ref{sec:meas:hash}, is capable of distilling entanglement for local depolarizing noise (LDN) up to $7 \%$ acting on each qubit of the resource state \cite{Zw14H}. This is due to an observation made in \cite{Zw13}: LDN acting on the input qubits of the resource state can {\it virtually} be moved to the initial states. Furthermore, LDN noise acting on the output qubits of the resource state can be assumed to act afterwards, since it commutes with the read-in measurements. These observations provide insights how one deals with LDN in MBQC, a physical characteristic which is not directly usable in quantum circuits, see Appendix \ref{sec:meas:hash}. More precisely, for gate-based implementations the situation is more complex and difficult to formalize in a useful way, since noise introduced by quantum gates gets highly correlated on propagating noise through the entire circuit. \newline
In a multipartite setting, a measurement-based implementation of the hashing protocol might turn out to be very useful for large scale quantum network architectures which rely on e.g. GHZ states \cite{Pirker17a}. \newline
In this paper we will use the terms {\it confidential}, {\it secure}, {\it privacy}, {\it private states} and {\it private entanglement}. Therefore we want to clarify their relationship and their distinction before using them. \newline
A communication channel, either classical or quantum, is referred to as confidential if an eavesdropper can not obtain any information regarding the data being transmitted. Nevertheless, the eavesdropper might change the data during transmission without being detected. Therefore we refer to privacy as the ability of two (or more) parties to establish a confidential communication channel. A communication channel is considered to be secure, if it is confidential and authenticated, where authenticated here means that the eavesdropper can not alter the data without being detected by the parties. In the quantum case we call a state private if it can be used to establish a confidential quantum channel, i.e., a state which is entangled between Alice and Bob but not entangled with the eavesdropper. The term private state was already introduced in the context of QKD for generating classical keys from states with bound entanglement \cite{Horodecki2005} and computing secret key capacities of quantum channels \cite{Pirandola2017}. For that purpose \cite{Horodecki2005,Pirandola2017} consider additional systems, known as shield systems, to decouple an eavesdropper from maximally entangled states to generate a secure key between two parties. However, privacy or private states as we consider here, refer to the ability of establishing a confidential quantum channel without the notion of shield systems. The entanglement of such a state is then referred to as private entanglement. \newline
For full formal definitions, proofs and supportive information we refer to the supplemental material. However, the confidentiality proof of hashing is self-contained in the main text. 

\section{Results}

We consider two categories of players: protocol participants and Eve, the eavesdropper, from which the participants request their initial states $\rho^{(n)}$ used for distillation. The former, connected via classical authenticated channels, wish to distill $m$ copies of a certain state $\ket{\varphi}$. In the bipartite setting, the state $\ket{\varphi}$ might correspond to a perfect Bell-pair \cite{Be96} whereas in the mutlipartite setting to a specific multipartite state \cite{Du05,EPPallGraphs,ChenEPPGraph,ManevaEPP,HostensEPPGraph,GlancyEPPGraph,HostensEPPGraph2} . The latter distributes the initial states via noisy quantum channels and has full control over them. In particular, Eve might be fully entangled with all initial states, which corresponds to the most general scenario how initial states can be distributed. \newline
Hashing in its original form assumes initial states of tensor product form, i.e. $\rho^{(n)} = \rho^{\otimes n}$ where $\rho$ is a density operator of a multi-partite quantum state and $n$ is asymptotically large. Furthermore, distillation will only be feasible if the entropy of the initial states is sufficiently low, see e.g. \cite{Be96} for bipartite hashing. \newline
To accommodate these requirements, we propose the following protocol: First the participants agree on a number of desired output systems $m$ and a confidentiality level $\varepsilon$. From these values they compute the number of systems $n$ which are necessary to meet both conditions, assuming the worst case entropy for the initial states. Then, the participants request $n + kn$ systems from Eve subject to distillation. They apply a local twirling operation which ensures that the systems are diagonal within the respective basis (for the bipartite protocol they twirl towards Werner form). Next, they sacrifice $kn$ systems for parameter estimation in order to estimate the actual fidelity $F$ relative to $\ket{\varphi}$ for each system. Depending on their estimate $\overline{F}$, they either abort the protocol because the fidelity is outside $[F_{\mathrm{min}}, F_{\mathrm{max}}]$ or they continue with a measurement-based implementation of the hashing protocol. Finally they output $m$ systems. When generalizing to arbitrary initial states the protocol will be prepended by a symmetrization step. \newline
To formalize our confidentiality criterion we recall some basic terminology introduced in \cite{Pirker16a}. We define the noiseless ideal map $\F$, which takes as input the initial states and outputs, depending on parameter estimation, either the asymptotic state of the hashing protocol, $\dm{\varphi}^{\otimes m}$, or some output state, $\sigma^{\perp}_{PE}$. For example, in a bipartite setting $\dm{\varphi}^{\otimes m} = \dm{B_{00}}^{\otimes m}$ where $\ket{B_{00}} = (\ket{00} + \ket{11})/\sqrt{2}$. The ideal map $\F$ abstracts the distillation protocol for an initial state $\rho$ as a process: internally it runs the real protocol for initial state $\rho$ to its very end which succeeds with probability $p_{\rho}$, and depending on parameter estimation, it either replaces the final state with its asymptotic state, or it outputs whatever state was reached by the protocol, $\sigma^{\perp}_{PE}$. This approach to define ideal functionality stems from well-established ideas in QKD \cite{RennerPost}. Formally we define
\begin{align}
(\F \otimes \id_E) & \left(\dm{\psi}_{PE} \right) = p_{\rho} \dm{\varphi}^{\otimes m} \otimes \sigma_E \otimes \dm{ok}_{f} \notag \\
& + (1-p_{\rho}) \sigma^{\perp}_{PE} \otimes \dm{fail}_{f} \label{eq:idealmap}
\end{align}
where $\ket{\psi}_{PE}$ is a purification of the initial state $\rho$ provided by Eve and $ p_{\rho}$ denotes the probability of the protocol succeeding for initial state $\rho$. The system $f$ distinguishes the accepting from the aborting branch. \newline
To analyze confidentiality taking into account realistic noisy scenarios, we also define the noisy ideal map $\F^\alpha$, where $\alpha$ characterizes the level of noise, as $\F^\alpha = \mathcal{N}^\alpha \circ \F$, where $\mathcal{N}^\alpha$ denotes the noise process acting on the output qubits of the resource states of hashing. \newline
We first clarify the noise processes we assume to act on the resource states of the measurement-based implementation of hashing, which motivate our definition of the ideal noisy map. We observe that there are a number of dominating sources of noise: noise on the resource states, noise on the read-in Bell measurements, and noise on the initial states subject to distillation. \newline
For the noise acting on the resource states we assume i.i.d. local depolarizing noise. This is physically reasonable due to the observations in \cite{Wallnofer16}, which shows that i.i.d. local depolarizing noise provides an accurate approximation of noise acting on resource states if these states get generated locally via entanglement distillation. \newline
The resource states for the measurement-based implementation of hashing consist only of input and output qubits, see Appendix \ref{sec:meas:hash} for further details. We denote the noise acting on the input qubits and output qubits of the resource states by $\mathcal{N}_{\mathrm{in}} = \prod^n_{j=1} \mathcal{D}_{j}(\alpha)$ and $\mathcal{N}_{\mathrm{out}} = \prod^m_{k=1} \mathcal{D}_{k}(\alpha)$ respectively where 
\begin{align}
\mathcal{D}_j(\alpha) \rho = \alpha \rho + \frac{1-\alpha}{4} (\rho + X_j \rho X_j + Y_j \rho Y_j + Z_j \rho Z_j) \label{eq:ldn},
\end{align}
with $\alpha \in [0, 1]$ quantifies the level of noise and the subscript $j$ denotes the qubit on which the Pauli operators act on. Furthermore, we can take into account for the noise which the read-in Bell measurements introduce by a lower value of $\alpha$ in $N_{\mathrm{in}}$, which we denote by $\beta$, see \cite{Zw14H}. Hence, we have $\mathcal{N}_{\mathrm{in}} = \prod^n_{j=1} \mathcal{D}_{j}(\beta)$. \newline
Because we can now mathematically shift the noise from the input qubits of the resource states to the initial states, we decompose the ideal noisy map $\F^\alpha$ as the concatenation of noise acting on the initial states followed by the noiseless ideal hashing protocol and noise acting on the  output qubits of the hashing protocol, i.e. $\F^\alpha = \mathcal{N}_{\mathrm{out}} \circ \F \circ \mathcal{N}_{\mathrm{in}}$. Because we can take into account for $\mathcal{N}_{\mathrm{in}}$ in the parameter estimation step of the ideal map $\F$ we end up with $\F^\alpha = \mathcal{N}^\alpha \circ \F$, where we have defined $\mathcal{N}^\alpha = \mathcal{N}_{\mathrm{out}}$. \newline
This enables us now to precisely define the term confidentiality. In particular, we call the hashing protocol $\mathcal{E}^\alpha$ $\varepsilon$-confidential, if
\begin{align}\label{eq:def:conf}
\| \E^\alpha  -\F^\alpha \|_{\diamond} \leq \varepsilon
\end{align}
where $\| \Delta \|_{\diamond} = \sup_{k \in \mathbb{N}} \|\Delta \otimes id_k \|_{\mathrm{op},1}$ for a CPTP map $\Delta$ with $\| \Delta \|_{\mathrm{op},1} := \sup_{\| \rho \|_1 \leq 1} \|\Delta (\rho)\|_1$ and $\| \rho \|_{1} = \tr{\sqrt{\rho \rho^\dagger }}$ denotes the $1-$norm of a density operator $\rho$, see also \cite{Kitaev1997}. \newline
Observe that the state $\dm{\varphi}^{\otimes m}$ in the accepting branch of $\F^\alpha$, see (\ref{eq:idealmap}), is private, i.e., a state which is disentangled from Eve. This motivates the term privacy distillation. \newline 
We outline the remainder of this paper as follows: We start by estimating the rate of convergence of noiseless bipartite hashing for finitely many i.i.d. initial states. Next, we generalize this result to arbitrary initial states including the eavesdropper's system via the post-selection technique. This will finally imply the confidentiality guarantees for the noisy measurement-based implementation of hashing. \newline
The hashing protocol \cite{Be96} deterministically converges exponentially fast towards several copies of $\ket{B_{00}}$ for i.i.d. initial states. In particular, we find for the modified (i.e., our proposed) hashing protocol $\E$, taking $n + kn$ initial states $\rho$, that
\begin{align}
\|\E & (\rho^{\otimes n + kn}) - \F (\rho^{\otimes n + kn}) \|_1 \notag \\
& \leq 2 \left[2 \exp(-n x_1(\delta)) + 2^{-n \delta} \right. \notag \\
& \left. + 2 \exp\left(-(F_{\mathrm{max}} - F_{\mathrm{min}})^2 kn/16 \right) \right] \label{eq:iid:finite}
\end{align}
where  $x_1(\delta) = 1/a_{\mathrm{max}} \left[ \left(g_{\mathrm{max}} + \delta \right) \log \left(1+\frac{\delta}{g_{\mathrm{max}}} \right) - \delta \right]$ and $a_{\mathrm{max}}$, $g_{\mathrm{max}}$ are  constants depending on $F_{\mathrm{min}}$ and $F_{\mathrm{max}}$. The parameter $\delta$ stems from the hashing protocol \citep{Be96} and affects the number of output systems $m = n(1-S(\rho)-2 \delta)$ where $S(\rho)$ denotes the von Neumann entropy of $\rho$ as well as the rate of convergence governed by (\ref{eq:iid:finite}). For our purposes we choose $\delta = n^{-1/5}$, see Appendix  \ref{app:bipartite:iid}. In addition, the right-hand side of (\ref{eq:iid:finite}) approaches zero  exponentially fast. \newline
Eq. (\ref{eq:iid:finite}) can be derived from the following observations, see also Appendix \ref{app:bipartite:iid}: The $1-$norm induced distance of $\E(\rho^{\otimes n + kn})$ and $\F (\rho^{\otimes n + kn})$ is equal to the distance within the $ok-$branch, because $\E$ and $\F$ agree on the $fail-$branch. The protocol can fail due to three reasons where each type of failure occurs with a certain probability. The first one corresponds to the case that the ensemble of Bell pairs falls outside of the likely subspace and is given by $2 \exp(-n x_1(n^{-1/5}))$. The second one bounds the probability of misidentifying the string by $\exp(-n^{4/5} \ln 2)$, and the third one bounds the failure probability of parameter estimation by $2 \exp\left(-(F_{\mathrm{max}} - F_{\mathrm{min}})^2 kn/16 \right)$. \newline
Nevertheless, (\ref{eq:iid:finite}) is insufficient to prove full cryptographic confidentiality, as it only concerns the systems of the participants and i.i.d. initial states. So the next step is to generalize (\ref{eq:iid:finite}) to arbitrary initial states including the system of Eve which is the topic of the next section. \newline
In order to provide an estimate of (\ref{eq:def:conf}) for bi- and multipartite hashing protocols in terms of i.i.d. initial states, e.g. (\ref{eq:iid:finite}), we proceed similar to the approach of \cite{Pirker16a}: First we relate the distance of the real and ideal map including Eve's purifying system at the beginning of the protocol to the distance between the respective maps concerning the systems of the participants only. Second we use the post-selection technique \cite{RennerPost}, which implies that the distance between the real and ideal map for any purification of the initial states is bounded by a specific pure state, a purification of the so called de-Finetti Hilbert-Schmidt state. \newline
We eliminate the first issue by using an inherent characteristic of noiseless distillation protocols: the target state of such protocols shared between Alice and Bob is pure, provided the parameter estimation is passed. Therefore the state of Alice and Bob is independent of Eve, i.e. there is no residual entanglement to her. We formalize this intuition via the following observation, rigorously proven in Appendix \ref{app:bipartite:localclose}: If the output of the real and ideal map, i.e. $\E$ and $\F$ respectively, differ at most $\varepsilon$ for a particular initial state $\rho$, then they differ at most $4 \sqrt{\varepsilon}$ on any purification $\ket{\psi}$ of $\rho$, i.e. 
\begin{align}
\| (\E \otimes id_{E} - \F \otimes id_{E})(\ketbra{\psi}{\psi}{ABE}) \|_1 \leq 4 \sqrt{\varepsilon}. \label{eq:localclose}
\end{align} \newline 
The next step is to relate non-i.i.d. initial states to i.i.d. initial states. Recall that the post-selection technique is applicable to permutation invariant maps only. Because hashing protocols are not permutation invariant maps, we have to prepend the overall protocol by a symmetrization step in order to apply the post-selection technique. This finally enables us to prove confidentiality of hashing protocols according to (\ref{eq:def:conf}) via the following theorem.
\begin{theo}[Post-selection-based reduction technique]\label{theo:postselect}
Let $\E^{s}$ be the real protocol and $\F^{s}$ the ideal protocol prepended by a symmetrization step ($s$) taking $n + kn$ initial states. Let $\E$ and $\F$ be the sub-protocols after symmetrization. Then we have 
\begin{align}
\|\E^{s} - \F^{s} \|_{\diamond} \leq 4 g_{n+kn,d} \sqrt{\max_{\sigma_{AB}} \left\|(\E-\F)\left(\sigma^{\otimes n + kn}_{AB} \right) \right\|_1} %\label{eq:thm:postselect}
\end{align}
where $d$ denotes the dimension of an individual system and $g_{n+kn,d} = {n+kn+d^2-1 \choose n} \leq (n+kn+1)^{d^2 - 1}$.
\end{theo}
The parameter $d$ in Theorem \ref{theo:postselect} corresponds to the dimension of each individual initial state, therefore it is constant for a specific protocol and we have for $M$ participants that $d = 2^M$. \newline
We sketch the proof of Theorem \ref{theo:postselect} as follows: The post-selection technique of \cite{RennerPost} implies that $\|\E^{s} - \F^{s}\|_{\diamond} = \sup_{\ket{\psi}_{ABE}} \|(\E^{s} - \F^{s}) \otimes id_{E}(\ketbra{\psi}{\psi}{}) \|_1$ is bound by evaluating this expression for a particular state, a purification of the de-Finetti Hilbert-Schmidt state. Hence we apply our previous observation, i.e. (\ref{eq:localclose}), to that particular initial state which reduces the confidentiality proof to i.i.d. initial states. For the complete proof of Theorem \ref{theo:postselect} we refer to Appendix  \ref{app:bipartite:thm}. \newline
We now easily conclude confidentiality of the noiseless bipartite hashing protocol prepended by symmetrization by combining Theorem \ref{theo:postselect} for $d=4$ and (\ref{eq:iid:finite}) which leads to 
\begin{align}
\|\E^{s} - \F^{s}\|_{\diamond} & \leq 4 \sqrt{2} (n+kn+1)^{15} \notag \\
& \times \left[2 \exp(-n x_1(n^{-1/5})) + \exp\left(-n^{4/5} \ln 2\right) \right. \notag \\
& \left. + 2 \exp\left(-(F_{\mathrm{max}} - F_{\mathrm{min}})^2 kn/16 \right) \right]^{1/2}.
\label{eq:conf:arb}
\end{align}
Eq. (\ref{eq:conf:arb}) analytically proves that arbitrary confidentiality levels can be achieved via the hashing protocol \cite{Be96} and finally enables us to show confidentiality for a noisy measurement-based implementation of the hashing protocol. \newline
Recall that the resource states, necessary for a measurement-based implementation of the hashing protocol, are subject to LDN acting on all qubits, $\mathcal{D}(\alpha) = \prod^n_{l=1} \mathcal{D}_l(\alpha)$ where $\mathcal{D}_l(\alpha)$ is defined in Eq. (\ref{eq:ldn}) and that we include the noise of a noisy Bell-measurement at the read-in in the value of $\alpha$ in (\ref{eq:ldn}), see \cite{Zw14H}. For a more detailed discussion of this noise model we refer to \cite{Wallnofer16} and Appendix \ref{sec:meas:hash}. \newline
The confidentiality proof for the noisy measurement-based implementation of hashing now concludes by using the following intuition from MBQC with resource states: the LDN on the input qubits can be moved, due to the symmetry of Bell-states, to the initial states whereas LDN acting on the output qubits can be assumed to act after the protocol. Therefore one is left with a noiseless hashing protocol generating pure states affected by LDN. We reiterate that such an approach is not directly applicable in the setting of gate-based implementations. \newline
We sharpen this observation as follows: The resource state of the protocol consists only of input and output qubits, see Appendix \ref{sec:meas:hash} and , and according to \cite{Zw13} we can virtually move the noise acting on the input qubits to the initial states provided by Eve. Thus we deal with this part of the noise via a modification of parameter estimation, since the entropy of the initial states increases after virtually moving the noise. The noise acting on the output qubits of the resource states can be assumed to act after the protocol completes, as that noise commutes with the read-in Bell-measurements. This leaves us with a noiseless protocol followed by LDN acting on the output qubits, which just slightly depolarizes the pure Bell-pairs from noiseless hashing. Moreover, this noise stems from the apparatus so this does not jeopardize confidentiality. In particular, because LDN is a CPTP map, the contractivity of the $1-$norm implies (see also Appendix \ref{app:noiseproof}) that
\begin{align}
\|\E^{s,\alpha} - \F^{s,\alpha} \|_{\diamond} \leq \|\E^{s} - \F^{s}\|_{\diamond} \label{eq:conf:noise:all}
\end{align} 
where $\E^{s,\alpha}$ and $\F^{s,\alpha}$ denote the real and the ideal noisy hashing protocol prepended by symmetrization, and noise of strength $1-\alpha$ of the form (\ref{eq:ldn}) acts on each qubit of the resource state independently and identically. Hence the noisy implementation offers the same confidentiality guarantees as the noiseless implementation, the protocol just simply aborts more often during parameter estimation. \newline
We highlight that the proof of confidentiality for noisy hashing does not require any numeric evidence, whereas the proof in \cite{Pirker16a} for the distillation protocol \cite{Deutsch} relies on numerical simulations. Furthermore the tolerable noise for post-selection is significantly higher, namely of the order of several percent per qubit compared to $O(10^{-20})$ in \cite{Pirker16a}, although it should be mentioned that the noise models are different and cannot directly be compared. \newline
Furthermore we find that there exists a regime of noise for bipartite hashing where privacy, or equivalently confidentiality, is achievable even though distillation is not feasible. For this regime, the privacy regime, hashing decreases the fidelity of each output system relative to $\ket{B_{00}}$, i.e., the protocol washes out entanglement rather than distilling it, but nevertheless, any eavesdropper factors out. In contrast, if the noise level is within the distillation regime the fidelity of each output system relative to $\ket{B_{00}}$ increases, and, as a consequence, any eavesdropper factors out. For private states in the context of QKD a similar observation was made in \cite{Horodecki2005}, where it was shown that even though entanglement distillation is not feasible yet secure keys can still be generated from private states with bound entanglement. 
 \newline 
It is interesting to qualitatively compare these findings to earlier work: in \cite{Aschauer02,As02b} confidentiality aspects were studied in the framework of a gate-based implementation of the entanglement distillation protocol of \cite{Deutsch}. It was also found that the noise regimes for privacy and distillation do not coincide, but contrary to the results presented here, the privacy regime for the gate based implementation was found to be a subset of the distillation regime. For more details on those noise regimes we refer to Appendix \ref{app:noiseregimes}. \newline
We consider the scenario where the local apparatus leaks all the information  about the noise processes realized (by the noisy resource states of the hashing protocol) to Eve as in \cite{Pirker16a,Aschauer02}. Theorem 7 of \cite{Pirker16a} states that if a real protocol $\E^\alpha$ is $\varepsilon$-confidential, then it is $2 \sqrt{\varepsilon}$-confidential if the noise transcripts leak to Eve. The resulting states remain private and enable for confidential quantum channels. \newline
The hashing protocol \cite{Be96} can be generalized to multipartite quantum states \cite{Du05,EPPallGraphs,ChenEPPGraph,ManevaEPP,HostensEPPGraph,GlancyEPPGraph,HostensEPPGraph2}, which is relevant for distributed quantum computation \cite{CiracDistributed}, quantum key agreement protocols \cite{XuQka,SunQka1,SunQka2} and quantum networks \cite{AcinNetwork,MeterNetwork,KimbleQInet,Pirker17a}. Also for those protocols one shows their confidentiality by following the same line of argumentation, which can be found in Appendix \ref{app:multi}. 

\section{Discussion} 

In summary we have analytically shown that noisy measurement-based implementations of bi- and multipartite hashing protocols establish exponential confidentiality levels. We directly exploited the properties of MBQC with resource states which leads, together with the purity of the asymptotic state of noiseless hashing and the post-selection technique, to a short, straightforward and transparent confidentiality proof. \newline
Furthermore, the privacy and distillation regimes do not coincide, similarly to private states with bound entanglement in the context of QKD. In particular, there exists a regime of local i.i.d. noise where privacy is achievable, but distillation is not. In this regime, any eavesdropper is factored out despite no entanglement being distilled. Nevertheless, in both regimes the final states are disentangled from any eavesdropper, which enables for secure quantum channels, if the information regarding the noise processes do not leak to the eavesdropper. If this information leaks to the eavesdropper, confidential quantum channels are still feasible as the resulting states remain private. \newline

%\section{Acknowledgements}

\begin{acknowledgements}

This work was supported by the Austrian Science Fund (FWF): P28000-N27, P30937-N27 and SFB F40-FoQus F4012, by the Swiss National Science Foundation (SNSF) through Grant number PP00P2-150579, the Army Research Laboratory Center for Distributed Quantum Information via the project SciNet and the EU via the integrated project SIQS.
\end{acknowledgements}

\appendix

\section{Bipartite hashing protocol and its measurement-based implementation}\label{app:bipartite:hashing}

In this section of the supplementary material we provide a short review of the biparite hashing protocol \cite{Be96}, we introduce the measurement-based implementation thereof \cite{Zw14H} and discuss its advantages over a gate-based approach. \newline
In the following we denote the four Bell-basis states by $\ket{B_{ij}} = (\id \otimes \px^j \pz^i) \ket{B_{00}}$ where $i \in \lbrace 0,1 \rbrace$ is referred to as the phase bit, $j \in \lbrace 0,1 \rbrace$ is referred to as the amplitude bit of $\ket{B_{ij}} $ and $\ket{B_{00}} = (\ket{00} + \ket{11})/\sqrt{2}$. 

\subsection{Entanglement distillation via hashing}\label{sec:epp}

Entanglement distillation protocols distill a maximally entangled state from several noisy copies provided the initial fidelity, defined as $F(\rho, \sigma) = \tr{\sqrt{\rho^{1/2} \sigma \rho^{1/2}}}$ for density operators $\rho$ and $\sigma$ where $\sigma = \dm{\varphi}$ (the desired target state), is sufficiently high. Several protocols have been proposed for this task, which we divide into two categories depending on the number of systems they utilize within each basic distillation step. In the first group we have recurrence-type protocols \cite{BennettRecurrence, Deutsch} 
which work pair-wise, whereas in the second group we have so-called hashing-type protocols \cite{Be96} that operate, in principle, on the entire ensemble. Common to both classes of protocols is that they utilize local operations, measurements and classical communication. \newline
Recurrence-type protocols are robust against local noise in both the gate-based \cite{DurRepeater} and measurement-based implementations \cite{Zw13}. In contrast, the gate-based implementations of hashing-type protocols are fragile with respect to noise of the local apparatus as we will discuss briefly. \newline 
The hashing protocol \cite{Be96} is an entanglement distillation protocol which operates on a large ensemble of noisy initial states in an iterative manner. In its standard version, the participants assume to receive $n$ copies of an initial state $\rho$, where $\rho$ is a two qubit density operator diagonal in the Bell-basis. The hashing protocol outputs $m=n(1-S(\rho))$ systems in the asymptotic limit where $S(\rho) < 1$ denotes the von-Neumann entropy of $\rho$. At each basic distillation step, which we also refer to as a round, the participants apply local operations according to a string drawn uniformly at random and followed by a controlled NOT into one target state. More precisely, they accumulate the phase and/or amplitude bit $i$ and $j$ of $\rho= \sum_{i,j} p_{ij} \dm{B_{ij}}$ of each individual pair into one target system via several controlled NOTs. Recall that such a bilateral controlled NOT transforms a tensor product of two Bell-states $\ket{B_{i_1 j_1}}$ and $\ket{B_{i_2 j_2}}$ to the tensor-product state $\ket{B_{i_1 \oplus i_2 j_1}} \ket{B_{i_1 j_1 \oplus j_2}}$. Next, the parties measure the target Bell-pair which is determined by the string. This measurement reveals essentially one bit of parity information about the remaining ensemble, thereby purifying it (as the mixedness of a state can be interpreted as a lack of classical information). The basic distillation step is iterated several times and in the end a fraction of purified systems remains. \newline
Hashing protocols rely on two fundamental concepts related to classical coding theory: likely subspace encoding and universal hashing. The idea of likely subspace encoding for ensembles of quantum states was first mentioned, to our knowledge, in \cite{Schumacher}. There it was proven that an asymptotic ensemble of i.i.d. quantum states $\rho^{\otimes n}$ where $\rho = \sum_{i} p_{i} \dm{v_i}$ is a density operator which receives almost all its weight from a small subspace spanned by so-called likely sequences $\{ \bigotimes_k |v_{i^{(j)}_k} \rangle \langle v_{i^{(j)}_k} | \}_{j \in J}$ where one identifies a specific sequence $\bigotimes_k \dm{v_{i_k}}$ with the bit string $(i_1,...,i_n)$. More precisely, the probability of finding a particular sequence $(j_1,...,j_n)$ that is outside this likely subspace can be made arbitrarily small in terms of the number of copies $n$ of $\rho$. In case of the hashing protocol the vectors $|v_{i} \rangle$ in $\rho = \sum_{i} p_{i} \dm{v_i}$ of the initial states $\rho^{\otimes n}$ correspond to individual Bell-states $\ket{B_{ij}}$. The original proposal of the likely subspace in \cite{Schumacher} relies on the weak law of large numbers, which is an asymptotic statement. Universal hashing \cite{UniversalHash} is a widely studied concept which turned out especially useful in privacy amplification \cite{BennettUniversal}, a critical part in quantum key distribution protocols. Privacy amplification minimizes the amount of information an eavesdropper has with respect to a generated key. For that purpose the participants use so-called $universal_2$ function families. A family of functions $\mathcal{G} = \lbrace g_i \colon A \to B \rbrace_{i \in I}$ is said to be $universal_2$ if for any $x \neq y \in A$ the probability that $g_i(x) = g_i(y)$ is at most $1/|B|$ when $g_i$ is chosen uniformly at random from $\mathcal{G}$. \newline
One basic distillation step of the hashing protocol comprises the following steps: one participant draws a string $s \in \{0,1,2,3\}^{n}$ (which we also refer to as parity hash string) uniformly at random, corresponding to a universal hash function. Next, the participant classically communicates $s$ to the other participant and both perform, according to $s$, local operations and bilateral controlled NOTs on their parts of the quantum states. Depending on $s_t \in \lbrace 0,1,2,3 \rbrace$ they bypass ($s_t = 0$) or they accumulate either the amplitude bit $j$ ($s_t = 1$), the phase bit $i$ ($s_t = 2$) or both, amplitude and phase bit $i \oplus j$, ($s_t = 3$) for the Bell-pair $\ket{B_{ij}}$ indexed by $1 \leq t \leq n$ into the first pair for which $s_t \neq 0$ via a bilateral controlled NOT. Finally, they measure both parts of this target system using the $Z$ observable which reveals almost one bit of parity information about the remaining ensemble. This basic distillation step is iterated $n-m$ times, thereby collecting sufficient amount of information regarding parities about the remaining quantum systems. The parity information is finally used to restore the systems to the $\ket{B_{00}}^{\otimes m}$ state. For further detsails on the hashing protocol, we refer the reader to \cite{Be96}. \newline
If one considers instead of asymptotic ensembles an initial ensemble of finite size $n$, bipartite hashing can still be used to distill entanglement. For finitely many initial states  slightly fewer systems with a finite infidelity (i.e. there is a non-zero deviation relative to the state $\ket{B_{00}}^{\otimes m}$) will be distilled. More precisely, for finite size hashing the number of output systems is $m = n(1-S(\rho) - 2 \delta)$ where the tunable parameter $\delta$ characterizes the width of the likely subspace. The parameter $\delta$ turns out to be crucial when determining the rate of convergence towards $\ket{B_{00}}^{\otimes m}$ and we will choose for our purposes $\delta = n^{-1/5}$. \newline
There also exist extensions of the bipartite hashing protocol to a multipartite setting allowing the distillation of two colorable graph states \cite{Du05}, all graph states \cite{EPPallGraphs}, GHZ states \cite{ChenEPPGraph,ManevaEPP}, CSS states \cite{HostensEPPGraph} and stabilizer states \cite{GlancyEPPGraph,HostensEPPGraph2}. Conceptually those types of protocols rely on the same ideas as bipartite hashing. Again, local parity collecting operations are used to reveal information about the remaining ensemble. They are especially well-suited to distill resource states for measurement-based implementations of particular quantum tasks such as quantum error correction. \newline
In the main text we have shown the confidentiality of the hashing protocol for two colorable graph states \cite{Du05} and we provide a detailed description thereof within this supplementary material.

\subsection{Measurement-based implementation}\label{sec:meas:hash}

One alternative to the gate-based implementation of a quantum circuit is measurement-based quantum computation \cite{BriegelMeas2,BriegelMeas}. A quantum operation $\mathcal{O}$ can be implemented by coupling the input qubits via Bell measurements to a universal resource state, e.g. a 2D cluster state \cite{BriegelCluster}. For circuits which contain only gates from the Clifford group and Pauli measurements one can also use an optimized, special purpose resource state of minimal size \cite{RausMeasQCcluster}. This resource state will consist of only $n+m$ qubits for a circuit which maps $n$ qubits to $m$ qubits. Hashing protocols, like most other entanglement distillation protocols, belong to this class of circuits and thus allow for such a minimal size measurement-based implementation. The results of the Bell measurements at the read-in determine both the results of the parity measurements of the hashing protocol as well as the Pauli byproduct operators on the final output states. For more informations and examples see \cite{Zw12,Pirker16a}.

The noiseless implementation of the hashing protocol produces asymptotically perfect Bell-pairs. Therefore any eavesdropper is factored out, in the limit, guaranteeing perfect confidentiality. But even if i.i.d. local depolarizing noise acts on the quantum gates, any gate-based approach fails \cite{Zw14H}. This is due to the $O(n)$ bilateral CNOTs within every distillation round, which washes out all information from the initial states. Hence the gate-based implementation of hashing is limited to the noiseless scenario only. \newline 
This drawback is overcome by a measurement-based approach \cite{Zw14H}. A measurement-based implementation of the hashing protocol is rather straightforward: a sequence of parity hash strings is drawn uniformly at random by one participant and classically communicated to all other participants. They construct the corresponding resource state according to that particular sequence. This resource state is finally coupled to the initial states via Bell-measurements which implements the hashing protocol in a measurement-based fashion. \newline
Since all gates of the hashing protocol are elements of the Clifford group the resource states consist only of input and output qubits, see discussion above. This implies that the resource states are of minimal size and therefore optimal with respect to the number of qubits which need to be stored temporarily. \newline
In \cite{Zw14H} it was shown that a measurement-based implementation of the hashing protocol \cite{Be96} is capable of distilling entanglement for imperfect resource states and imperfect in-coupling Bell-measurements. There the resource states are affected by i.i.d. local depolarizing noise (LDN) of the form $\mathcal{D}(\alpha) = \prod^n_{l=1} \mathcal{D}_l(\alpha)$ acting on all qubits of the resource states where  
\begin{align}
\mathcal{D}_j(\alpha) \rho = \alpha \rho + \frac{1-\alpha}{4} (\rho + X_j \rho X_j + Y_j \rho Y_j + Z_j \rho Z_j) \label{eq:ldn:app}
\end{align}
and $\alpha$ characterizes the strength of the noise. In particular, the measurement-based implementation of hashing tolerates up to $7 \%$ of noise acting on each qubit of the resource state \cite{Zw14H}. In \cite{DurStandardNoise}, it was shown that any local noise process can be brought into a local depolarizing form. This observation also motivated the noise model of local depolarizing noise chosen in \cite{Zw13} to study measurement-based recurrence-type distillation protocols. There it was shown that the measurement-based implementation of recurrence-type distillation protocols is capable of tolerating up to $24 \%$ of noise acting on each qubit of the resource state. Furthermore, as studied in \cite{Wallnofer16}, local i.i.d. depolarizing noise provides an accurate and reasonable approximation if one generates the resource states via entanglement distillation. The generation of resource states via entanglement distillation also provides an efficient scheme to create high-fidelity resource states, crucial for accurate measurement-based quantum computation via resource states. \newline
The reason why a measurement-based implementation of the hashing protocol in the presence of i.i.d. LDN of the form $\mathcal{D}(\alpha)$ works is due to a fundamental observation made in \cite{Zw13}: If the resource states undergo a local depolarizing noise of the form $\mathcal{D}(\alpha) = \prod^n_{l=1} \mathcal{D}_l(\alpha)$ then one can {\it virtually} exchange the location of the LDN when followed by a Bell-measurement, i.e. $\mathcal{P} \mathcal{D}_1(\alpha) \rho = \mathcal{P} \mathcal{D}_2(\alpha) \rho$ where $\mathcal{P} \rho = P_B \rho P^\dagger_B$ and $P_B$ denotes a projector on a Bell-state. Intuitively speaking, as $P_B = \dm{B_{ij}}$, this is due the symmetry $(\id \otimes \sigma) \ket{B_{ij}} = (\sigma \otimes \id) \ket{B_{ij}}$ up to a global phase where $\sigma$ is a Pauli operator. This enables us to effectively move the noise acting on the input qubits of the resource states to the input state (as we couple the input state to the resource state via Bell-measurements). We emphasize that this holds for LDN of the form $\mathcal{D}(\alpha) = \prod^n_{l=1} \mathcal{D}_l(\alpha)$ and, more importantly, this can not be done within the circuit model even though the gate-based and measurement-based approach to quantum computation are computationally equivalent. In particular, computational equivalence does not necessarily imply equivalent robustness with respect to noise. This observation becomes more clear when one considers the noise processes as being part of the protocol. In the measurement-based scenario with resource states, the observation of \cite{Zw13} implies that the i.i.d. LDN acting on the input qubits of the resource state can effectively be moved to the initial states, see discussion above. The i.i.d. LDN acting on the output qubits can be applied afterwards, because the quantum computation at hand is performed in terms of Bell-measurements at the read-in. This leaves one with a perfect quantum operation on a modifed input state, where i.i.d. LDN is applied, followed by the noise process of the output qubits. In \cite{Zw15} this observation was applied to measurement-based quantum communication, where it was shown that very high error thresholds (of the order of 10 \% per qubit) can be obtained. In contrast, in the gate-based approach noise accumulates through repeatedly applying quantum gates. Furthermore, on commuting noise through the gates of a quantum circuit towards the input, the noise processes might get correlated due to commutation relations, maybe ending up in correlated noise rather than i.i.d. LDN acting on the input state. So to summarize, this observation shows that at least for i.i.d. LDN the measurement- and gate-based approach are not equivalent. \newline
To summarize, the measurement-based approach permits a noisy implementation of the hashing protocol whereas a standard gate-based implementation fails in the presence of noise.

\section{Noise regimes}\label{app:noiseregimes}

In the main text we identified two different regimes of i.i.d. LDN of the form $\mathcal{D}(\alpha) = \prod^n_{l=1} \mathcal{D}_l(\alpha)$, where $\mathcal{D}_l(\alpha)$ is defined via (\ref{eq:ldn:app}), acting on the resource states of the measurement-based implementation of hashing: privacy and purification regime. Within the first regime any eavesdropper factors out but no entanglement will be distilled. In particular, for bipartite hashing, the fidelity relative to $\ket{B_{00}}$ will decrease due to the protocol. In contrast, in the purification regime any eavesdropper is factored out and entanglement is distilled, i.e. the fidelity relative to the target state increases. \newline
To see this we recall the conditions on the noise parameters for purification and privacy. The noiseless hashing protocol distills perfect Bell pairs in the asymptotic limit of infinitely many initial states in Werner form as soon as their fidelity exceeds $F_{\mathrm{crit}} = 0.8107$, see \cite{Be96}. In this case the final Bell pairs are private (and thus confidentiality is guaranteed) and $F_{\mathrm{crit}}$ can be translated to $q_{\mathrm{crit}} = (4 F_{\mathrm{crit}} - 1)/3 \approx 0.7476$. In the noisy case one has two conditions for the noise parameters $\alpha$ and $q$, which quantify the level of noise on the resource states and the fidelity of the initial states, respectively (see also \cite{Zw13}) for asymptotic ensemble sizes:
\begin{align}
\label{sec}
\alpha^2q > q_{\mathrm{crit}}
\end{align}
and
\begin{align}
\label{pur}
\alpha^2 > q.
\end{align}
Here, (\ref{sec}) guarantees that the fidelity of the initial states, after the noise from the resource state is mapped to the initial states, see the previous section and \cite{Zw13}, exceeds the threshold value $q_{\mathrm{crit}}$. In this case the output pairs will be private. The second condition, (\ref{pur}), ensures that the fidelity of the output pairs is larger than the fidelity of the input pairs.
From this one sees that for privacy one only needs to fulfill (\ref{sec}), whereas both (\ref{sec}) and (\ref{pur}) need to hold for purification. Observe that (\ref{sec}) is a condition due to the noise acting on the input qubits (thereby increasing the required fidelity of the initial states to succeed hashing) whereas condition (\ref{pur}) stems from the noise applied to the output qubits (which depolarizes the perfect Bell-pairs produced by noiseless hashing in the asymptotic limit). This means that the parameters $\alpha$ and $q$ are more constrained if one aims for increasing entanglement, as compared to the case of privacy. We summarize these findings in Fig. \ref{fig:secreg}. \newline
This observation provides a clear distinction between privacy and purification regime for asymptotic ensembles: Both regimes, purification and privacy, have in common that any eavesdropper factors out due to the protocol but they differ with respect to whether entanglement is distilled or not. This motivates the term quantum privacy distillation for the proposed overall protocol as there are noise regimes where the protocol offers privacy, or equivalently private entanglement, without achieving distillation.

A similar situation arises in the finite size case. Here, the modifications will be that $q_{\mathrm{crit}}$ in (\ref{sec}) is no longer directly related to $F_{\mathrm{crit}}$ and that (\ref{pur}) needs to be modified to
\begin{align}
\alpha^2 q_{\mathrm{out}}(n,F) > q_{\mathrm{in}}.
\end{align}
Here, $q_{\mathrm{out}}(n,F)$ quantifies the level of noise on the output pairs of the hashing protocol for $n$ initial states with fidelity $F$. It can be obtained from the bound on the fidelity of the output pairs. There will again be two different regimes, and the purification regime will be smaller than the privacy regime due to the fact that it is more constrained (there are two inequalities to be satisfied, whereas there is only one for confidentiality).

\begin{figure}[h!]
\scalebox{0.4}{
\includegraphics{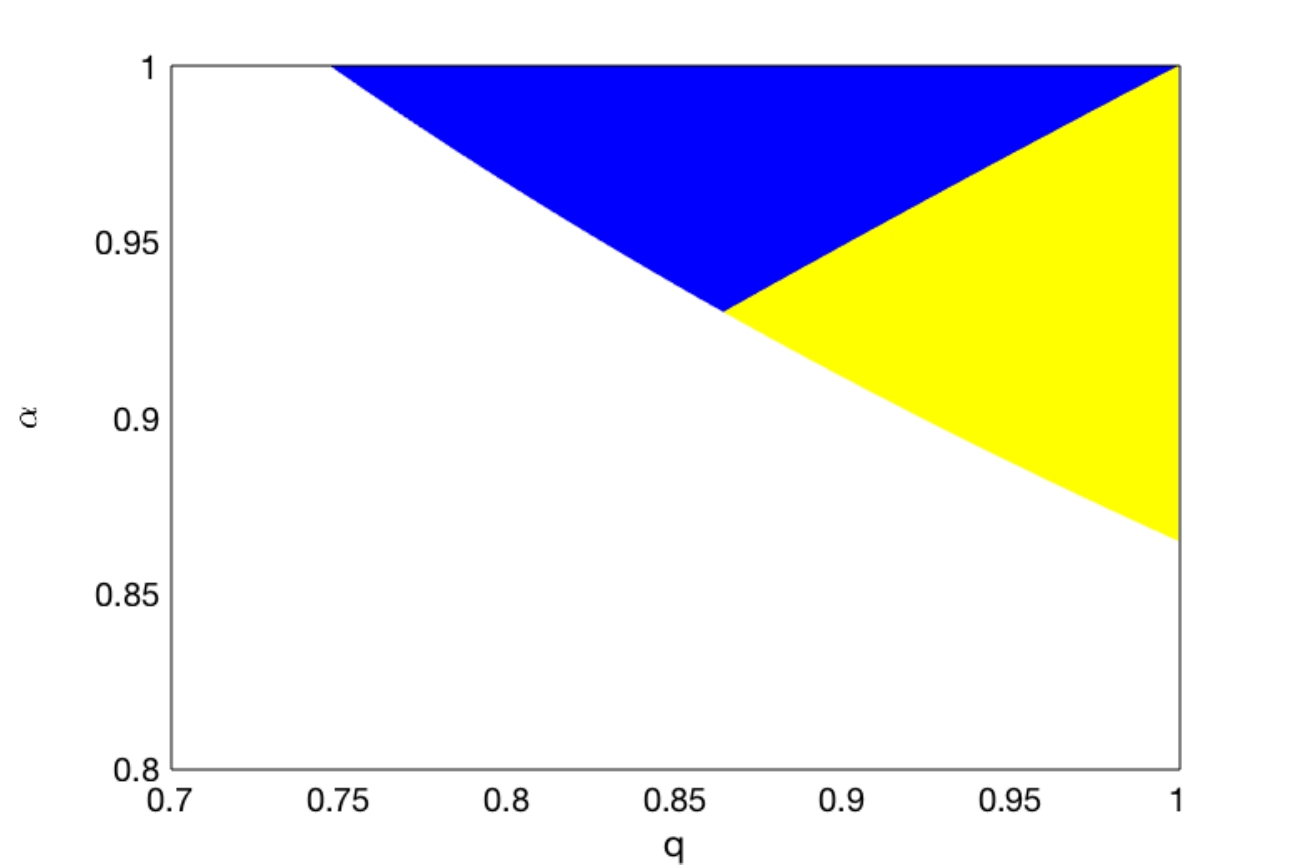}
}
%\flushleft
\caption[h!]{\label{fig:secreg} Visualization of the different regimes in the $\alpha - q$ plane (only the upper right corner of the entire plane is shown). In the white area neither privacy nor purification is achieved. In the entire colored area privacy is guaranteed, but only in the blue area one has distillation. This means that there is a parameter regime (yellow area), where one has privacy despite the fact that the fidelity of the Bell pairs does not increase during the distillation.}
\end{figure}

\section{Rate of convergence of noiseless bipartite hashing for i.i.d. initial states}\label{app:bipartite:iid}

Here we provide the proof of Eq. (\ref{eq:iid:finite}) of the main text for $\delta = n^{-1/5}$ summarized within the following Theorem. 
\begin{theo} [Convergence for i.i.d. initial states]
Let $\E$ be the real protocol and $\F$ the ideal protocol taking $n + kn$ initial states. Furthermore, let $x_1(\delta) = 1/a_{\mathrm{max}} \left[ \left(g_{\mathrm{max}} + \delta \right) \log \left(1+\frac{\delta}{g_{\mathrm{max}}} \right) - \delta \right]$ where $a_{\mathrm{max}}$ and $g_{\mathrm{max}}$ are  constants depending on $F_{\mathrm{min}}$ and $F_{\mathrm{max}}$. Then we have for all initial states $\rho$ that
\begin{align}
\|\E & (\rho^{\otimes n + kn}) - \F (\rho^{\otimes n + kn}) \|_1 \notag \\
& \leq 2 \left[2 \exp(-n x_1(n^{-1/5})) + \exp\left(-n^{4/5} \ln 2\right) \right. \notag \\
& \left. + 2 \exp\left(-(F_{\mathrm{max}} - F_{\mathrm{min}})^2 kn/16 \right) \right]. 
\end{align}
Furthermore, the right-hand side of Eq. (\ref{eq:iid:finite}) of the main text approaches exponentially fast zero.
\end{theo}
\begin{proof}
Because the ideal and the real map are identical in the aborting branch, we find for the initial states $\rho^{\otimes n + kn}$ that 
\begin{align}
\|\E (\rho^{\otimes n + kn}) & - \F (\rho^{\otimes n + kn}) \|_1 \notag \\
& = p_{\rho} \|\sigma_{AB} - \dm{B_{00}}^{\otimes m} \| \leq \varepsilon_{\text{H}} \label{eq:conf:local}
\end{align}
where $\sigma_{AB}$ denotes the state of the hashing protocol after $n-m$ rounds and $p_{\rho}$ the success probability for initial state $\rho$. Thus we need to estimate $\varepsilon_{\text{H}}$. Because we twirl the initial states towards Werner form we assume from now on that they are of Werner form. \newline
The hashing protocol can fail due to two reasons, see  \cite{Be96}: the string corresponding to the initial states falls outside the likely subspace or, after $n-m$ rounds two or even more configurations are compatible with the total parity information, i.e. they can not be distinguished from each other. \newline
By denoting this failure probabilities by $p'_1$ and $p'_2$ and the corresponding states after the protocol by $\sigma_1$ and $\sigma_2$ respectively, we find that the total failure probability $p'_f$ of the hashing protocol satisfies $p'_f = p'_1 + p'_2$. We also observe that if the parameter estimation was accurate the state after the protocol completes, i.e. $\sigma_{AB}$ of (\ref{eq:conf:local}), is given by  
\begin{align}
\sigma_{AB} = (1-p'_f) \dm{B_{00}}^{\otimes m} + \sum^2_{i=1} p'_i \sigma_i. \label{eq:conf:local:bi:finalstate}
\end{align}
More precisely, with probability $1-p'_f$ we are able to restore the output of the hashing protocol to $m$ copies of $\ket{B_{00}}$ and we end up with probabilities $p'_1$ and $p'_2$ in the state $\sigma_1$ and $\sigma_2$ respectively. This implies for (\ref{eq:conf:local}) that 
\begin{align}
\|\sigma_{AB} - \dm{B_{00}}^{\otimes m}\|_1 \leq 2(p'_1 + p'_2) \label{eq:conf:local:ok:1}
\end{align}
via the triangle inequality for the case whenever parameter estimation is accurate. \newline
Additionally the overall protocol can fail due to the following observation: The parameter estimation provides an estimate $\overline{F}$ for the fidelity $F$ which is accepted by the participants, but $F$ is actually outside the agreed range $[F_{\mathrm{min}},F_{\mathrm{max}}]$. In that case Alice and Bob run hashing even though the protocol will either fail (since the initial fidelity is too low) or the fidelity is too high to provide accurate confidentiality estimates \footnote{The hashing protocol requires $F > F_{\mathrm{crit}}$ where $F_{\mathrm{crit}} = 0.8107$ to distill entanglement from the initial states. The restriction that $F < F_{\mathrm{max}}$ is due to the applicability of Bennett's inequality which requires bounded random variables. However, for the noisy implementation of the hashing protocol this criterion will be met automatically as the resource states for the measurement-based implementation undergo an i.i.d. LDN process.}. This observation in turn implies that the state after hashing within the ok-branch is maximum far from the asymptotic state of the hashing protocol, i.e. 
\begin{align}
\|\sigma_{AB} - \dm{B_{00}}^{\otimes m}\|_1 \leq 2. \label{eq:pefail:state:ok}
\end{align}
Nevertheless, the probability of the protocol succeeding for initial state $\rho$ also takes into account for  parameter estimation succeeding, i.e. $p_{\rho} = p'_3 \cdot p'$ where $p'_3$ denotes the probability of parameter estimation succeeding for initial state $\rho$. Therefore, if Alice and Bob mistakenly run hashing even if they should have aborted we find via (\ref{eq:pefail:state:ok}) for (\ref{eq:conf:local}) that 
\begin{align}
p_{\rho} \|\sigma_{AB} - \dm{B_{00}}^{\otimes m}\|_1 \leq 2 p'_3. \label{eq:pefail:ok}
\end{align}
So to summarize we obtain for an arbitrary initial state $\rho$ by combining (\ref{eq:conf:local:ok:1}) and (\ref{eq:pefail:ok}) that 
\begin{align}
p_{\rho} \|\sigma_{AB} - \dm{B_{00}}^{\otimes m}\|_1 \leq 2(p'_1 + p'_2 + p'_3). \label{eq:conf:local:ok:2}
\end{align}
Thus we are left to provide upper bounds for (the unknown) probabilities $p'_1$, $p'_2$ and $p'_3$ respectively, i.e. we need to find $p_1$, $p_2$ and $p_3$ such that $p'_i \leq p_i$ for $1 \leq i \leq 3$ because this implies for (\ref{eq:conf:local:ok:2}) that 
\begin{align}
p_{\rho} \|\sigma_{AB} - \dm{B_{00}}^{\otimes m}\|_1 \leq 2(p_1 + p_2 + p_3). \label{eq:conf:local:ok}
\end{align}
 \newline
We derive a bound for the probability of falling outside the likely subspace $p_1$ via the Bennett inequality \cite{BennettInequ}. Bennett's inequality \cite{BennettInequ}  states that we have for $X_1,..,X_n$ independent random variables, where $|X_i| \leq a$ almost-surely and the expected value of $X_i$ is zero w.l.o.g., that
\begin{align}
\mathrm{Pr} \left(\left|\sum^n_{i=1} X_i \right| > t \right) \leq 2 \exp\left( - \frac{n \sigma^2}{a^2} h\left(\frac{at}{n\sigma^2} \right) \right) \label{eq:bennett}
\end{align}
where $\sigma^2 = 1/n \sum^n_{i=1} \mathrm{Var} X_i$ and $h(u)=(1+u)\log(1+u)-u$ \footnote{Observe that Bennett's inequality is only applicable to bounded random variables, which is also the reason why we propose to accept only initial states where the fidelity is within an agreed range $[F_{\mathrm{min}}, F_{\mathrm{max}}]$.}. \newline
For the hashing protocol the random variables $X_i$ take the values  $X_i(k,l) := -\log_2 p_{kl} - S(\rho)$ where $\rho = \sum^1_{k,l=0} p_{kl} |B_{kl}\rangle \langle B_{kl}|$ and $S(\rho) = - \sum^1_{k,l=0} p_{kl} \log_2 p_{kl}$ denotes the von-Neumann entropy. The von-Neumann entropy simplifies for states in Werner form to $S(\rho) = -F \log_2 (F) - (1-F) \log_2 ((1-F)/3) =: S(F)$. \newline
The i.i.d. assumption implies that all $X_i$ are independent and identical distributed (therefore we will subsequently denote them by the random variable $X$), thus we find $\sigma^2 = 1/n \sum^n_{i=1} \mathrm{Var} X_i = \mathrm{Var} X =: V(F)$. Hence we have 
\begin{align}
V(F) &= \mathrm{Var} X = \sum_{k,l} p_{kl} (-\log_2 p_{kl} - S(F))^2 \notag \\
&= \sum_{k,l} p_{kl} (\log^2_2 p_{kl} + 2 S(F) \log_2 p_{kl} + S^2(F)) \notag \\
&= \sum_{k,l} p_{kl} \log^2_2 p_{kl} + 2 S(F) p_{kl} \log_2 p_{kl} + p_{kl} S^2(F) \notag \\
&= \sum_{k,l} p_{kl} \log^2_2 p_{kl} + 2 S(F)(-S(F)) + S^2(F) \notag \\
&= F \log^2_2 F + (1-F) \log^2_2 ((1-F)/3) - S^2(F).
\end{align}
We observe that the random variable $X$ is bounded. More precisely, we have $|X(k,l)| = |\log_2 p_{kl} + S(F)| \leq |\log_2((1-F)/3)| + S(F) =: a(F)$ because $|\log_2 ((1-F)/3)| > |\log_2 F|$ for $F > 0.8107$ (which is the minimum required fidelity for Werner states by the hashing protocol). \newline
The next step is to insert $t=n\delta$, $a = a(F)$ and $\sigma^2 = V(F)$ in  (\ref{eq:bennett}) which yields by denoting the left-hand-side of (\ref{eq:bennett}) by $p_1$
\begin{align}
& p_1 \leq 2 \exp \left(\frac{-n V(F)}{a^2(F)} h \left(\frac{a(F)n \delta}{n V(F)} \right) \right) \notag  \\
&= 2 \exp \left \lbrace \frac{-n V(F)}{a^2(F)} \left[ \left(1 + \frac{a(F) \delta}{V(F)} \right) \log \left(1+\frac{a(F) \delta}{V(F)} \right) \right. \right. \notag  \\
& \left. \left. - \frac{a(F) \delta}{V(F)} \right] \right \rbrace \notag \\ 
&= 2 \exp \left \lbrace \frac{-n}{a(F)} \left[ \left(\frac{V(F)}{a(F)} + \delta \right) \log \left(1+\frac{a(F) \delta}{V(F)} \right) - \delta \right] \right \rbrace.
\end{align}
By defining $g(F) = \frac{V(F)}{a(F)}$ we rewrite the previous inequality as
\begin{align}
p_1 \leq 2 \exp \left \lbrace \frac{-n}{a(F)} \left[ \left(g(F) + \delta \right) \log \left(1+\frac{\delta}{g(F)} \right) - \delta \right] \right \rbrace. \label{eq:likely:outside:1}
\end{align}
We observe that (\ref{eq:likely:outside:1}) depends on the fidelity $F$ of the initial states which is inappropriate for confidentiality estimates. In order to obtain a bound which is independent of the fidelity of the initial states we use that Alice and Bob only run the hashing protocol if $F \in [F_{\mathrm{min}}, F_{\mathrm{max}}]$. \newline %For that purpose we rewrite (\ref{eq:likely:outside:1}) as
%\begin{align}
%p_1 \leq 2 & \exp \left( \frac{n \delta}{a(F)} \right) \notag \\
%&  \exp\left[ \frac{-n }{a(F)} \left(g(F) + \delta \right) \log \left(1+\frac{\delta}{g(F)} \right) \right] \label{eq:likely:outside:2}
%\end{align}
We observe that (\ref{eq:likely:outside:1}) is maximized whenever $\frac{n}{a(F)} [\left(g(F) + \delta \right) \log \left(1+\frac{\delta}{g(F)} \right) -\delta]$ is minimal because $\left(g(F) + \delta \right) \log \left(1+\frac{\delta}{g(F)} \right) - \delta \geq 0$ which follows from $\log(1+x) \geq \frac{x}{x+1}$, $n > 0$ and $a(F) > 0$. \newline
For that purpose we show that the function $y(x)=(x+b)\log(1+b/x) - b = (x+b)(\log(x+b) - \log(x)) - b$ is strictly monotonic decreasing in $x$. We obtain for the first derivative of $y$ that
\begin{align}
y'(x) &= \log(x+b) + \frac{x+b}{x+b} - \left(\log(x) + \frac{x+b}{x} \right) \notag \\
& = \log(x+b) + 1 - \log(x) - 1 - \frac{b}{x} \notag \\
& = \log\left(\frac{x+b}{x}\right) - \frac{b}{x} = \log\left(1+\frac{b}{x}\right) - \frac{b}{x} \leq 0
\end{align}
since $\log(1+z) \leq z$. Thus $\left(g(F) + \delta \right) \log \left(1+\frac{\delta}{g(F)} \right) - \delta \to \mathrm{min}$ whenever $g(F) \to \mathrm{max}$. 
\begin{figure}[h!]
\includegraphics[width=\columnwidth]{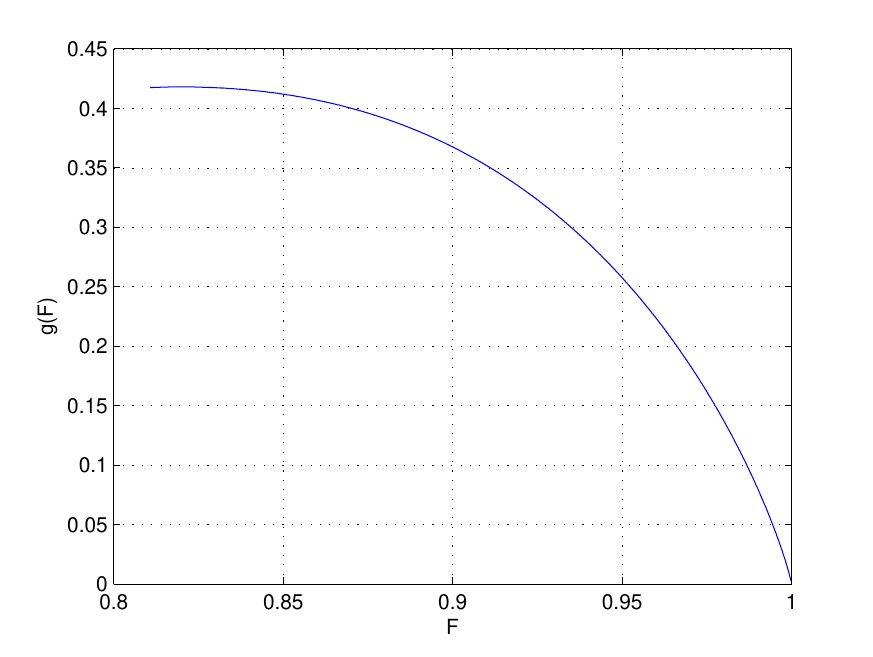} 
\caption[h!]{Plot of the function $g(F)$. Observe that $g$ is strictly monotonic decreasing for $F \in [0.82, 1)$.}\label{fig:g}
\end{figure} \newline
From Figure \ref{fig:g} we see that $g(F) \to \mathrm{max}$ for $F \to \mathrm{min}$. This implies for (\ref{eq:likely:outside:1}) that 
\begin{align}
p_1 \leq 2 \exp & \left \lbrace \frac{-n}{a(F)} \right. \notag \\
& \left. \left[ \left(g(F_{\mathrm{min}}) + \delta \right) \log \left(1+\frac{\delta}{g(F_{\mathrm{min}})} \right) - \delta \right] \right \rbrace. 
 \label{eq:likely:outside:3}
\end{align}
Consequently $\left(g(F_{\mathrm{min}}) + \delta \right) \log \left(1+\frac{\delta}{g(F_{\mathrm{min}})} \right) - \delta \geq 0$ and $a(F) \leq a(F_{\mathrm{max}})$ implies 
\begin{align}
p_1 \leq 2 \exp \left \lbrace \frac{-n}{a_{\mathrm{max}}} \left[ \left(g_{\mathrm{max}} + \delta \right) \log \left(1+\frac{\delta}{g_{\mathrm{max}}} \right) - \delta \right] \right \rbrace \label{eq:likely:outside:4}
\end{align} 
where $a_{\mathrm{max}} = a(F_{\mathrm{max}})$ and $g_{\mathrm{max}} = g(F_{\mathrm{min}})$. We rewrite (\ref{eq:likely:outside:4}) in a more compact form by defining $x_1(\delta) = 1/a_{\mathrm{max}} \left[ \left(g_{\mathrm{max}} + \delta \right) \log \left(1+\frac{\delta}{g_{\mathrm{max}}} \right) - \delta \right]$ and inserting $\delta=n^{-1/5}$ \footnote{The choice of $\delta$ is a trade-off between the rate of convergence and the number of output system $m = n(1-S(\rho)-2\delta)$. Any choice $\delta < n^{-1/4}$ is appropriate.} as 
\begin{align}
p_1 \leq 2 \exp(-n x_1(n^{-1/5})). \label{eq:likely:outside}
\end{align}
We will use (\ref{eq:likely:outside}) for the confidentiality estimate (\ref{eq:conf:local:ok}). In order to show that (\ref{eq:likely:outside}) ensures an exponential convergence, as we claim, we need to provide an upper bound for the exponent of (\ref{eq:likely:outside}), i.e., for the function
\begin{align}
-\frac{n}{a_{\mathrm{max}}} \left[ \left(g_{\mathrm{max}} + \delta \right) \log \left(1+\frac{\delta}{g_{\mathrm{max}}} \right) - \delta \right] \label{eq:expconv:1}
\end{align} 
where $\delta$ will be choosen later as $n^{-1/5}$ as previously. By defining 
\begin{align}
f(n) = n \left((g_{\mathrm{max}}+\delta)\log\left(1+ \frac{\delta}{g_{\mathrm{max}}}\right) - \delta \right) \label{eq:expconv:f}
\end{align}
(\ref{eq:expconv:1}) reads as $-f(n)/a_{\mathrm{max}}$. In the following we compute a lower bound $y(n)$ for $f(n)$, i.e. $f(n) > y(n)$ for all $n$, which is in turn an upper bound for (\ref{eq:likely:outside}), i.e. $p_1 \leq 2 \exp(-f(n)/a_{\mathrm{max}}) \leq 2 \exp(-y(n)/a_{\mathrm{max}})$. Using that $\log(1+x) > \frac{x}{1 + x/2}$ for $x>0$, see \cite{Love}, we find from $g_{\mathrm{max}} > 0$ and $\delta>0$ that 
\begin{align}
\log\left(1+ \frac{\delta}{g_{\mathrm{max}}}\right) > \frac{\frac{\delta}{g_{\mathrm{max}}}}{1+\frac{\delta}{2g_{\mathrm{max}}}} = \frac{2\delta}{2g_{\mathrm{max}} + \delta}. \label{eq:expconv:2}
\end{align}
Furthermore we have that $(g_{\mathrm{max}}+\delta)\log\left(1+ \frac{\delta}{g_{\mathrm{max}}}\right) - \delta \geq 0$ which implies together with (\ref{eq:expconv:2}) for (\ref{eq:expconv:f})
\begin{align}
f(n) & >  n \left((g_{\mathrm{max}}+\delta)\frac{2\delta}{2g_{\mathrm{max}} + \delta} - \delta \right) \notag \\
& = n \frac{2\delta (g_{\mathrm{max}}+\delta) - (2g_{\mathrm{max}} + \delta) \delta}{2g_{\mathrm{max}} + \delta} \notag \\
& = n \frac{2\delta g_{\mathrm{max}} + 2 \delta^2 - 2g_{\mathrm{max}} \delta - \delta^2}{2g_{\mathrm{max}} + \delta} \notag \\
& = \frac{n \delta^2}{2g_{\mathrm{max}} + \delta} \geq \frac{n \delta^2}{2g_{\mathrm{max}} + 1}
\end{align}
because $\delta \leq 1$. Inserting $\delta = n^{-1/5}$ finally gives
\begin{align}
f(n) >  \frac{n \delta^2}{2g_{\mathrm{max}} + 1} = \frac{n^{3/5}}{2g_{\mathrm{max}} + 1} =: y(n)
\end{align}
implying
\begin{align}
p_1 &\leq 2\exp(-f(n)/a_{\mathrm{max}}) \leq 2\exp(-y(n)/a_{\mathrm{max}}) \notag \\
& = 2\exp\left(-\frac{n^{3/5}}{a_{\mathrm{max}}(2g_{\mathrm{max}} + 1)} \right)
\end{align}
which analytically proves the exponential scaling of the hashing protocol. \newline
Furthermore, following the approach of \cite{Be96}, we find that the probability of having two configurations which are compatible with the collected parity information, $p_2$, is bounded by $2^{-n \delta}$. Thus, inserting $\delta = n^{-1/5}$ gives $p_2 < 2^{- n^{4/5}}$. \newline
Finally we provide an estimate for the probability of accepting initial states from Eve in the case when Alice and Bob should abort the protocol after parameter estimation, i.e. the actual fidelity $F$ is below the minimum required value $F_{\mathrm{min}}$ but the estimate $\overline{F}$ is not, or the actual fidelity $F$ is above $F_{\mathrm{max}}$ but the estimate $\overline{F}$ is not, corresponding to the probability $p'_3$. For that purpose we perform two-qubit measurements of two Bell-pairs, the first w.r.t. the $X \otimes X$ and the second w.r.t. the $Z \otimes Z$ observable. One easily observes that $\ket{B_{00}}$ is the common $+1$ eigenstate of both operators. By referring to this measurements as $M_1$ and $M_2$ respectively and recalling that the parameter estimation utilizes $kn$ systems we define the random variables $F_i$ associated with a pair of Bell-pairs for $1 \leq i \leq kn/2$ which is equal to $1$ whenever $M_1$ and $M_2$ simultaneously reveal outcome $1$ and $0$ otherwise. \newline
Recall that the Hoeffding inequality \cite{Hoeffding} states that we have for $X_1,..,X_n$ i.i.d. random variables where $a_i \leq X_i \leq b_i$, $c_i = b_i - a_i$, $S_n = \sum_i X_i$ and the expected value $E_n$ of $S_n$, i.e. $E_n = \mathrm{E} [S_n]$, that
\begin{align}
\P \left(|S_n - E_n| > t \right) < 2 \exp\left( - \frac{2t^2}{n C^2} \right) \label{eq:hoeffding}
\end{align}
holds for all $t$ and where $\forall i: c_i \leq C$. Hoeffding's inequality (\ref{eq:hoeffding}) implies now for the empirical mean $\overline{F} = 2/(kn) \sum^{kn/2}_{i=1} F_i$ that
\begin{align}
\mathrm{Pr}(|\overline{F} - \mathrm{E}[F]| > \eta) < 2 \exp(-\eta^2 kn) \label{eq:pefail:eta}
\end{align}
holds for all $\eta$. More precisely, the probability of estimating an error larger than $\eta$ via $\overline{F}$ to $\mathrm{E}[F]$ is decaying exponential in $n$. So Alice and Bob choose $F_{\mathrm{min}}$ and $F_{\mathrm{max}}$ and they agree to continue with the hashing protocol whenever $\overline{F} \in [F_{\mathrm{PE}} - \Delta/4, F_{\mathrm{PE}} + \Delta/4]$ where $F_{\mathrm{PE}} = (F_{\mathrm{min}} + F_{\mathrm{max}})/2$ and $\Delta = F_{\mathrm{max}} - F_{\mathrm{min}}$. Fixing $\eta = \Delta/4$ implies for (\ref{eq:pefail:eta}) that
\begin{align}
\mathrm{Pr}(|\overline{F} - \mathrm{E}[F]| > \eta) < 2 \exp(-(F_{\mathrm{max}} - F_{\mathrm{min}})^2 kn/16). \label{eq:pefail}
\end{align}
In other words, (\ref{eq:pefail}) means that the probability that Alice and Bob continue with the hashing protocol in case they should abort, i.e., the actual fidelity $F$ is outside $[F_{\mathrm{min}}, F_{\mathrm{max}}]$, is exponentially small. For example, if the fidelity estimate $\overline{F}$ is $\overline{F} = F_{\mathrm{PE}} + \Delta/4$ (which implies Alice and Bob will run hashing), then the probability that the actual fidelity $F$ satisfies $F > F_{\mathrm{PE}} + \Delta/2 = F_{\mathrm{max}}$ is exponentially bounded. \newline
To summarize, we find for (\ref{eq:conf:local}) that
\begin{align}
\|\E & (\rho^{\otimes n + kn}) - \F (\rho^{\otimes n + kn}) \|_1 \notag \\
& \leq 2 \left[2 \exp(-n x_1(n^{-1/5})) + \exp\left(-n^{4/5} \ln 2\right) \right. \notag \\
& \left. + 2 \exp\left(-(F_{\mathrm{max}} - F_{\mathrm{min}})^2 kn/16 \right) \right]. \label{eq:iid:finite:1}
\end{align}
Notice that the right-hand side of (\ref{eq:iid:finite:1}) is independent of $\rho$, which completes the proof.
\end{proof}

\section{Local closeness implies global closeness}\label{app:bipartite:localclose}

In the main text we formulated the following claim: If the output of the real and ideal map differ at most $\varepsilon$ for a particular initial state then they differ at most $4 \sqrt{\varepsilon}$ for any purification of this initial state. We prove this statement within the following Lemma.
\begin{lemma}\label{lem:localcloseness}
Let $\E$ be the real and $\F$ be the ideal protocol. Furthermore let $\rho$ be a mixed state shared by the participants of the protocol. If $\|\E(\rho) - \F(\rho)\|_1 \leq \varepsilon$, then
\begin{align}
\| (\E \otimes id_{E} - \F \otimes id_{E})(\ketbra{\psi}{\psi}{ABE}) \|_1 \leq 4 \sqrt{\varepsilon}
\end{align}
for all purifications $\ket{\psi}_{ABE}$ of $\rho$.
\end{lemma}
\begin{proof}
We observe that 
\begin{align}
\E(\rho) & = p_\rho \sigma_{AB} \otimes \ketbra{ok}{ok}{} & \notag \\
& + (1-p_\rho) \sigma^{\perp}_{AB} \otimes \ketbra{fail}{fail}{} & \\
\F(\rho) & = p_\rho \dm{\varphi}^{\otimes m}_{AB} \otimes \ketbra{ok}{ok}{} & \notag \\
& + (1-p_\rho) \sigma^{\perp}_{AB} \otimes \ketbra{fail}{fail}{} . &
\end{align}
The assumption $\|\E(\rho) - \F(\rho)\|_1 \leq \varepsilon$ implies $p_{\rho} \|\sigma_{AB} - \dm{\varphi}^{\otimes m}_{AB}\|_1 \leq \varepsilon$ because $\E(\rho)$ and $\F(\rho)$ are equal on the fail branch. Thus we have $\|\sigma_{AB} - \dm{\varphi}^{\otimes m}_{AB} \|_1 \leq \varepsilon / p_{\rho}$.

Furthermore we find for the application of the real and the ideal protocol to the purification $\ket{\psi}_{ABE}$ of $\rho_{AB}$ that
\begin{align}
(\E \otimes id_{E}) & (\dm{\psi}_{ABE}) = p_{\rho} \sigma_{ABE} \otimes \ketbra{ok}{ok}{} & \notag \\
& + (1-p_{\rho}) \sigma^{\perp}_{ABE} \otimes \ketbra{fail}{fail}{}, & \\
(\F \otimes id_{E}) & (\dm{\psi}_{ABE}) = p_{\rho} \dm{\varphi}^{\otimes m}_{AB} \otimes \rho_{E} \otimes \ketbra{ok}{ok}{} & \notag \\
& + (1-p_{\rho}) \sigma^{\perp}_{ABE} \otimes \ketbra{fail}{fail}{}. &
\end{align}
This implies for the $1$-norm that
\begin{align}
\|(\E \otimes id_{E} &- \F \otimes id_{E})(\dm{\psi}_{ABE}) \|_1 \notag \\
& = p_\rho \|\sigma_{ABE} -  \dm{\varphi}^{\otimes m}_{AB} \otimes \rho_{E} \|_1. \label{eq:global:1}
\end{align}
Thus we need to show that $p_\rho \|\sigma_{ABE} - \dm{\varphi}^{\otimes m}_{AB} \otimes \rho_{E}\|_1 \leq 4 \sqrt{\varepsilon}$.  One easily verifies $\ptr{E}{\sigma_{ABE}} = \sigma_{AB}$ and $\ptr{AB}{\sigma_{ABE}} = \rho_{E}$ because the system $E$ is not affected by the protocol $\E$. Recall that we have by assumption that $\|\sigma_{AB} - \dm{\varphi}^{\otimes m}_{AB} \|_1 \leq \varepsilon / p_{\rho}$. Thus we apply Lemma 10 of the supplementary material from \cite{Pirker16a} to $\rho_{SE} := \sigma_{ABE}$ and $\varphi_{SE} := \dm{\varphi}^{\otimes m}_{AB} \otimes \rho_{E}$ where $S := AB$ which implies
\begin{align}\label{equ:global:2}
\|\sigma_{ABE} - \dm{\varphi}^{\otimes m}_{AB} \otimes \rho_{E}\|_1 \leq 4 \sqrt{\varepsilon / p_{\rho}}.
\end{align}
Employing (\ref{equ:global:2}) in (\ref{eq:global:1}) yields
\begin{align}
\|(\E \otimes id_{E} & - \F \otimes id_{E})(\ketbra{\psi}{\psi}{ABE}) \|_1 \notag \\
& \leq p_\rho 4 \sqrt{\varepsilon / p_{\rho}} = 4 \sqrt{p_{\rho} \varepsilon} \leq 4 \sqrt{\varepsilon}
\end{align}
which completes the proof.
\end{proof}

\section{Proof of Theorem \ref{theo:postselect}}\label{app:bipartite:thm}

%\begin{theorep}
%Let $\E^{s}$ be the real protocol and $\F^{s}$ the ideal protocol prepended by a symmetrization step ($s$) taking $n + kn$ initial states. Let $\E$ and $\F$ be the sub-protocols after symmetrization. Then we have 
%\begin{align}
%\|\E^{s} - \F^{s} \|_{\diamond} \leq 4 g_{n+kn,d} \sqrt{\max_{\sigma_{AB}} \left\|(\E-\F)\left(\sigma^{\otimes n + kn}_{AB} \right) \right\|_1} %\label{eq:thm:postselect}
%\end{align}
%where $d$ denotes the dimension of an individual system and $g_{n+kn,d} = {n+kn+d^2-1 \choose n} \leq (n+kn+1)^{d^2 - 1}$.
%\end{theorep}
\begin{proof}
Due to the symmetrization we find that $\E^{s}$ and $\F^{s}$ are permutation invariant maps. Hence applying the post-selection technique of \cite{RennerPost} gives 
\begin{align}
\|\E^{s} & - \F^{s} \|_{\diamond} \notag \\
& \leq g_{n+kn,d} \|(\E^{s} \otimes id_{E} - \F^{s} \otimes id_{E})(\ketbra{\tau}{\tau}{ABE})\|_1 \label{equ:post:1}
\end{align}
where $d$ is determined by the number of participants (see discussion below) and $\ket{\tau}_{ABE}$ is a purification of the de-Finetti Hilbert-Schmidt state, hence $\ptr{E}{\ketbra{\tau}{\tau}{ABE}} = \int \sigma^{\otimes n+kn}_{AB} d\mu(\sigma) =: \tau'$ where $\mu$ is the measure induced by the Hilbert-Schmidt metric on $\text{End}(\mathbb{C}^d)$. One easily observes that 
\begin{align}
\|\E^{s}(\tau') & - \F^{s}(\tau')\|_1 = \left\|(\E^{s} - \F^{s})\left(\int \sigma^{\otimes n + kn}_{AB} d\mu(\sigma) \right) \right\|_1 \notag \\
& \leq  \max_{\sigma_{AB}} \left\|(\E-\F)\left(\sigma^{\otimes n + kn}_{AB} \right) \right\|_1 \label{eq:post:2}
\end{align}
where $\E$ and $\F$ denote the subprotocols after symmetrization. As $\ket{\tau}_{ABE}$ is a purification of $\tau'$ we can apply Lemma \ref{lem:localcloseness} implying for (\ref{equ:post:1}) that
\begin{align}
\|\E^{s} & - \F^{s} \|_{\diamond} \notag \\
& \leq g_{n+kn,d} \|(\E^{s} \otimes id_{E} - \F^{s} \otimes id_{E})(\ketbra{\tau}{\tau}{ABE})\|_1 \notag \\
& \leq g_{n+kn,d} 4 \sqrt{\|(\E^{s} - \F^{s})(\ptr{E}{\ketbra{\tau}{\tau}{ABE}})\|_1 } \notag \\
& = 4g_{n+kn,d} \sqrt{\|(\E^{s} - \F^{s})(\tau')\|_1 } \notag \\
& \leq 4 g_{n + kn,d} \sqrt{\max_{\sigma_{AB}} \left\|(\E-\F)\left(\sigma^{\otimes n + kn}_{AB} \right) \right\|_1}
\end{align}
where the second inequality stems from Lemma \ref{lem:localcloseness} and the last inequality from (\ref{eq:post:2}) which finally shows the claim. 
\end{proof}

\section{Confidentiality of a noisy measurement-based implementation of the hashing protocol}\label{app:noiseproof}

Within this section we prove Eq. (\ref{eq:conf:noise:all}) of the main text. In doing so, we formulate the following Theorem.
\begin{theo}
Let $\E^{s,\alpha}$ and $\F^{s,\alpha}$ be the real and the ideal noisy hashing protocol prepended by symmetrization where noise of strength $1-\alpha$ of the form (\ref{eq:ldn:app}) acts on each qubit of the resource state independent and identical. Then
\begin{align}
\|\E^{s,\alpha} - \F^{s,\alpha} \|_{\diamond} \leq \|\E^{s} - \F^{s}\|_{\diamond}. 
\end{align}
\end{theo}
\begin{proof}
The resource state necessary for the measurement-based implementation of hashing is pure and minimal in the number of qubits and consists only of input and output qubits, because all quantum gates involved in the hashing protocol are elements of the Clifford group \cite{RausMeasQCcluster}. \newline
Hence there are only two different locations at which noise acts: input and output qubits. For the noise acting on the input qubits we use the observation made in \cite{Zw13}, which enables us to {\it virtually} move the noise from the input qubits to the initial states, thereby increasing their entropy. For the noise acting on the output qubits, as described in the main text, we can safely assume that this noise will act after the protocol completes, leaving us with a noiseless hashing protocol (w.r.t. the output qubits). \newline
We deal with the noise on the input qubits by a slight modification of the parameter estimation step. Recall that Alice and Bob fix $F_{\mathrm{min}}$ and $F_{\mathrm{max}}$ for parameter estimation and they continue with the hashing protocol whenever their fidelity estimate $\overline{F}$ is within the interval $[F_{-},F_{+}]$ where $F_{\pm} = F_{\mathrm{PE}} \pm \Delta / 4$ for $F_{\mathrm{PE}} = (F_{\mathrm{max}} + F_{\mathrm{min}})/2$ and $\Delta = F_{\mathrm{max}} - F_{\mathrm{min}}$. The noise acting on the input qubits of the resource states increases the entropy of the initial states which forces Alice and Bob to accept less initial states from Eve. By describing the initial states in an i.i.d. setting after the twirl via i.i.d. LDN of the form (\ref{eq:ldn:app}), i.e. $\rho = D_1(q) \dm{B_{00}}$, the parameter estimation interval $[F_{-},F_{+}]$ transforms to  $[q_{-},q_{+}]$ via $q_{\pm} = (4F_{\pm} - 1)/3$. According to the previous observation that we can virtually move the noise of level $\alpha$ on the input qubits of the resource states, $D_1(\alpha)$ and $D_2(\alpha)$ respectively, to the initial states we consequently describe the initial states as $D_2(\alpha)D_1(\alpha)D_1(q) \dm{B_{00}} = D_1(\alpha^2) D_1(q) \dm{B_{00}} = D_1(\alpha^2 q) \dm{B_{00}}$, see also Fig. \ref{fig:app:noise:input}. Observe that we have moved the noise from Bob's to Alice's side due to the symmetry of Bell-states. 
\begin{figure}[h!]
\centering
\includegraphics[scale=1]{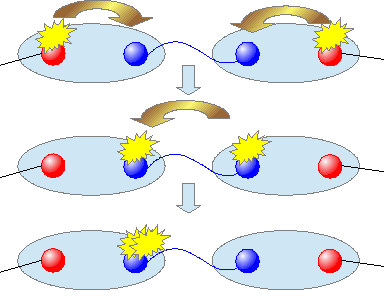}  
 %\subfloat[\centering]{\includegraphics[width=\columnwidth]{yield_fmin_085_087_n-15.eps} \label{fig:confidentiality:yield}} 
\caption[h!]{\label{fig:app:noise:input} The figure shows, at an abstract level, how noise on the input qubits is moved from the resource states to the initial states. The blue ellipsis indicate Bell-measurements, the red vertex the input qubits of the resource state and the light-blue vertex the qubits of the initial states.}
\end{figure}
Thus we need to have $\alpha^2 q \in [q_{-},q_{+}]$ to pass the parameter estimation and run the hashing protocol. Observe that $\alpha^2 q$ transforms to the fidelity $F'$ of the initial states, including the noise of the resource state, via $\alpha^2 q = (4F'-1)/3$. Therefore we modify the parameter estimation to continue with the hashing protocol whenever the estimate of the fidelity $\overline{F}$ of the initial states satisfies 
\begin{align}
\overline{F} \in \left[\frac{3 q_{-} + \alpha^2}{4 \alpha^2},\frac{3 q_{+} + \alpha^2}{4 \alpha^2} \right], \label{eq:pemodified}
\end{align}
see Fig. \ref{fig:app:noise:pe}. \newline
\begin{figure}[h!]
\centering
\includegraphics[width=\columnwidth]{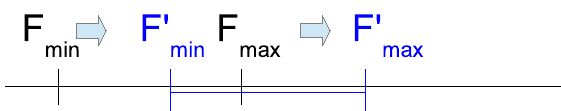}  
 %\subfloat[\centering]{\includegraphics[width=\columnwidth]{yield_fmin_085_087_n-15.eps} \label{fig:confidentiality:yield}} 
\caption[h!]{\label{fig:app:noise:pe} The interval for parameter estimation acceptance $[F_{\mathrm{min}},F_{\mathrm{max}}]$ transforms according to (\ref{eq:pemodified}) to $[F'_{\mathrm{min}},F'_{\mathrm{max}}]$.}
\end{figure}
We denote the protocols with modified parameter estimation according to condition (\ref{eq:pemodified}) by the maps $\E^{s,\alpha-\mathrm{in}}$ and $\F^{s,\alpha-\mathrm{in}}$ respectively. It follows immediately from the definition of the protocols that we achieve the same confidentiality level of Eq. (\ref{eq:conf:arb}) of the main text as for the noiseless protocols, Alice and Bob will just abort the protocol more often. Hence we easily deduce 
\begin{align}
\|\E^{s,\alpha-\mathrm{in}} - \F^{s,\alpha-\mathrm{in}}\|_{\diamond} = \|\E^{s} - \F^{s}\|_{\diamond}. \label{eq:conf:noise:in}
\end{align}
We now extend the confidentiality proof to a full noisy measurement-based implementation of the hashing protocol protocol as follows: Since we can effectively move noise of level $\alpha$ acting on the input qubits of the resource states to the to-be-purified ensemble, the modification (\ref{eq:pemodified}) of the parameter estimation extends the confidentiality proof via (\ref{eq:conf:noise:in}) to noise acting on the input qubits of the resource state. For noise acting on the output qubits we use the following observation: Because the noise is assumed to be of the form (\ref{eq:ldn:app}) it is also CPTP. By denoting the noise acting on the output qubits as $\NN^{\alpha} = \bigotimes^{m}_{j=1} \mathcal{D}_{j;A}(\alpha) \mathcal{D}_{j;B}(\alpha)$ where $A$ and $B$ denote Alice's and Bob's parts of the final Bell-pairs, the noisy real protocol and ideal protocol read as $\E^{s,\alpha} = \NN^{\alpha} \circ \E^{s,\alpha-\mathrm{in}}$ and $\F^{s,\alpha} = \NN^{\alpha} \circ \F^{s,\alpha-\mathrm{in}}$ respectively \footnote{This is due to the fact that we can assume that the noise acting on the output qubits is applied after the protocol(s) have finished.}. Hence (\ref{eq:conf:noise:in}) and the contractivity of the $1-$norm for CPTP maps imply  
\begin{align}
\|\E^{s,\alpha} - \F^{s,\alpha} \|_{\diamond} & = \| \NN^{\alpha} \circ (\E^{s,\alpha-\mathrm{in}}-\F^{s,\beta-\mathrm{in}}) \|_{\diamond} \notag \\
& \leq  \| \E^{s,\alpha-\mathrm{in}} - \F^{s,\alpha-\mathrm{in}} \|_{\diamond} = \|\E^{s} - \F^{s}\|_{\diamond}.
\end{align}
What remains to be dealt with are the Pauli byproduct operators due to the measurement outcomes at the inputs, but since LDN of the form (\ref{eq:ldn:app}) commutes with the Pauli byproduct operators we do not have to worry about them in the proof of confidentiality, which completes the proof.
\end{proof}

\section{Confidentiality of multiparty hashing protocol for two-colorable graph states}\label{app:multi}

We start by recalling some basic notation, definitions and properties of graph states. \newline
We define the graph state basis $\ket{\psi_{\kappa_1,\ldots,\kappa_{N}}}$ where $\kappa_1,\ldots,\kappa_{N} \in \lbrace 0,1 \rbrace$ associated with a graph $G=(V,E)$ where $N = |V|$ as the common eigenstate of the correlation operators 
\begin{align}
K_{j} = X^{(j)} \prod_{\lbrace j,k \rbrace \in E} Z^{(k)} \label{eq:graph:correlation}
\end{align}
with eigenvalues $(-1)^{\kappa_j}$ for $1 \leq j \leq N$ where the superscript denote the qubit on which the Pauli operator is acting on. We refer to the state $\ket{\psi_{0,\ldots,0}}$ also as the graph state associated with $G=(V,E)$. Note that the states $\lbrace \ket{\psi_{\kappa_1,\ldots,\kappa_{N}}} \rbrace^1_{\kappa_1,\ldots,\kappa_{N} = 0}$ form a basis of the Hilbert-space $(\mathbb{C}^{2})^{\otimes N}$. A special class of graph states are so-called two-colorable graph states which correspond to two-colorable graphs. A graph is said to be two-colorable if there exists a mapping $f \colon V \to \lbrace 1, 2 \rbrace$ such that for all vertices $v \in V$ it holds that $f(v) \neq f(w)$ for all neighbors $w \in V$ of $v$. The most prominent example of  two-colorable graph states are GHZ and cluster states \cite{BriegelCluster}. \newline
Suppose we want to distill a two-colorable graph state $\ket{\psi_{0 \ldots 0}}$ corresponding to a graph $G=(V,E)$ where $V = V_A \cup V_B$, $A$ and $B$ denote the colors and $|V_A| = N_A$, $|V_B| = N_B$ where $N = N_A + N_B$. The multipartite hashing protocol assumes asymptotically many i.i.d. initial states $\rho$ diagonal in the graph state basis, i.e. $\rho = \sum_{\bm{\mu}, \bm{\nu}} \lambda_{\bm{\mu}, \bm{\nu}} \dm{\psi_{\bm{\mu}, \bm{\nu}}}$ where $\bm{\mu} = (\mu_1,\ldots,\mu_{N_A}) \in \lbrace 0,1 \rbrace^{N_A}$ and $\bm{\nu} = (\nu_1,\ldots,\nu_{N_B}) \in \lbrace 0,1 \rbrace^{N_B}$ are multi-indices corresponding to color $A$ and $B$ respectively \footnote{If the initial states are not diagonal in the graph state basis we achieve this by probabilistically applying the correlation operators (\ref{eq:graph:correlation}), see \cite{Du05}. This procedure is also referred to as twirling.}. \newline
For two-colorable graph states we define multilateral CNOTs on two copies $\rho_1$ and $\rho_2$ which enable us to transfer information between the initial states $\rho_1$ and $\rho_2$. More precisely, by applying a CNOT to all particles in $V_A\,(V_B)$ where $\rho_1$ serves as target(source) and $\rho_2$ as source (target) a straightforward computation leads to (by denoting this unitary as $U_1$)
\begin{align}
\ket{\psi_{\bm{\mu}, \bm{\nu}}} \otimes \ket{\psi_{\bm{\mu'}, \bm{\nu'}}} \stackrel{U_{1}}{\to} \ket{\psi_{\bm{\mu}, \bm{\nu} \oplus \bm{\nu'}}} \otimes \ket{\psi_{\bm{\mu} \oplus \bm{\mu'}, \bm{\nu'}}}. \label{eq:mcnot:1}
\end{align}
By exchanging the roles of $V_A$ and $V_B$ one obtains (by denoting this unitary as $U_2$)
\begin{align}
\ket{\psi_{\bm{\mu}, \bm{\nu}}} \otimes \ket{\psi_{\bm{\mu'}, \bm{\nu'}}} \stackrel{U_{2}}{\to} \ket{\psi_{\bm{\mu} \oplus \bm{\mu'}, \bm{\nu}}} \otimes \ket{\psi_{\bm{\mu'}, \bm{\nu} \oplus \bm{\nu'}}} \label{eq:mcnot:2}.
\end{align}
Suppose we measure all qubits of the graph state $|\psi_{\mu_1,\ldots,\mu_{N_A},\nu_1,\ldots,\nu_{N_B}} \rangle$ belonging to the set $V_A$ with the $X$ and all qubits of the set $V_B$ with the $Z$ observable. By denoting the outcomes of the $X$ measurements with $\xi_i \in \lbrace 0,1 \rbrace$ and the outcomes of the $Z$ measurements with $\zeta_j \in \lbrace 0,1 \rbrace$ one immediately finds via (\ref{eq:graph:correlation}) 
\begin{align}
\mu_i = \left( \xi_i + \sum_{\lbrace i,j \rbrace \in E} \zeta_j \right) \mod 2
\end{align} 
for all $1 \leq i \leq N_A$. In other words, we can use this measurement setting to reveal information about all $\kappa_i$ for $1 \leq i \leq N_A$ simultaneously. We refer to this measurements with $M_{1}$. Similarly, by exchanging the roles of $V_A$ and $V_B$ we obtain information about all $\nu_i$ for $1 \leq i \leq N_b$. In the following, we refer to this measurements with $M_{2}$. \newline
The multiparty hashing protocol is now defined as follows \cite{Du05}: In order to reveal information about color $A$, i.e. $\bm{\mu}$, (which we denote as sub-protocol $P_1$) we apply $U_1$ to a random subset of the $n$ initial states with common target system (thereby accumulating the values corresponding to color $A$) and perform measurement $M_1$ on this common system. Similarly, by applying $U_2$ to a random subset of the initial states with a common target system (thereby accumulating the values corresponding to color $B$) followed by $M_2$ on this common system one obtains information about color $B$, i.e. $\bm{\nu}$ (which we denote as sub-protocol $P_2$). Repeating the sub-protocols $P_1$ and $P_2$ sufficiently many times leads to perfect knowledge about the remaining states, i.e. one ends up in a pure state (which we restore to the target state $|\psi_{0,\ldots,0} \rangle^{\otimes m}$). \newline
Recall that the overall protocol prepends the multiparty hashing protocol by a twirling and parameter estimation step. The twirling step ensures that the initial states are diagonal within the graph state basis, see \cite{Du05}, whereas the participants use parameter estimation to decide whether the multiparty hashing protocol will succeed or not. \newline
Formally, we define the probabilities 
\begin{align}
a^{(\mu_i)}_{i} &= \sum_{\mu_k \neq \mu_j, \nu} \lambda_{\mu_1, \ldots, \mu_i, \ldots, \mu_{N_A}, \nu} \label{eq:multi:iid:a} \\
b^{(\nu_j)}_{j} &= \sum_{\nu_k \neq \nu_j, \mu} \lambda_{\mu,\nu_1, \ldots, \nu_j, \ldots, \nu_{N_B}} \label{eq:multi:iid:b}
\end{align}
for $1 \leq i \leq N_A$ and $1 \leq j \leq N_B$. For example, for a three-qubit state we have $a^{(0)}_{1} = \sum_{k,l} \lambda_{0kl}$ and $a^{(1)}_{1} = \sum_{k,l} \lambda_{1kl}$. Observe that the values $S(a_{i})$ and $S(b_{j})$ correspond to the entropies of $\mu_i$ and $\nu_j$ within the vectors $\bm{\mu}$ and $\bm{\nu}$. \newline 
As shown in \cite{Du05}, the protocol described above is in the asymptotic limit capable of distilling $m=n(1-\max_{1 \leq i \leq N_A} S(a_{i}) - \max_{1 \leq j \leq N_B} S(b_{j}))$ copies of the state $|\psi_{0,\ldots,0} \rangle$.  \newline
Now we are ready to compute the distance of the real and ideal multiparty hashing protocol for i.i.d. initial states. Intuitively it follows from the same arguments as in the bipartite setting. 
\begin{theo}
Let $\E$ be the real and $\F$ be the ideal multiparty hashing protocol. Furthermore let $\rho$ be an initial state. Then
\begin{align}
\|\E (\rho^{\otimes n + kn}) & - \F (\rho^{\otimes n + kn}) \|_1 \leq \varepsilon_{\text{H}} \label{eq:conf:multi:local} 
\end{align}
where $\varepsilon_{\text{H}} \in O(\exp(-\sqrt{n}))$ is independent of the initial state $\rho$. 
\end{theo}
\begin{proof}
Recall that the multiparty hashing protocol aims to distill several copies of a two-colorable graph state via the sub-protocols $P_1$ for color $A$ and $P_2$ for color $B$ from $n$ copies of the initial state $\rho = \sum_{\bm{\mu},\bm{\nu}} \dm{\psi_{\bm{\mu}, \bm{\nu}}}$ where the states $\ket{\psi_{\bm{\mu}, \bm{\nu}}}$ correspond to the graph state basis. \newline
The crucial observation is that we learn the values of $\bm{\mu}$ and $\bm{\nu}$ corresponding to the colors $A$ and $B$ within $n$ copies of the initial state $\rho = \sum_{\bm{\mu},\bm{\nu}} \dm{\psi_{\bm{\mu}, \bm{\nu}}}$ via the sub-protocols $P_1$ and $P_2$ independently. In other words, $\bm{\mu}$ and $\bm{\nu}$ do not get correlated during the protocol execution, i.e. they remain independent. By taking a closer look at $P_1$ $(P_2)$ we infer that also the individual components of $\bm{\mu}$ $(\bm{\nu})$ remain independent. In particular, the components of $\bm{\mu} = (\mu_1,\ldots,\mu_{N_A})$ $(\bm{\nu} = (\nu_1,\ldots,\nu_{N_B}))$ remain distinct during the protocol, i.e. for each $i$ the value $\mu_i$ is independent of $\mu_k$ for all $k \neq i$ (for each $j$ the value $\nu_j$ is independent of $\nu_k$ for all $k \neq j$). This is due to the fact that $U_1$ $(U_2)$ operates component-wise on $\bm{\mu}$ $(\bm{\nu})$ \footnote{Intuitively speaking this independence stems from the two-colorability of the graph-state and the properties of $U_1$ and $U_2$.}. \newline
Keeping this observations in mind, it is straightforward to provide finite size estimates for the fidelity of the state after the protocol relative to $\ket{\psi_{0,\ldots,0}}$. Observe that the hashing protocol fails if either $P_1$ or $P_2$ fails which implies for the failure probability $p_f$ of the hashing protocol $p_f \leq p_{P_1} + p_{P_2}$ where $p_{P_1}$ and $p_{P_2}$ denote the failure probabilities of sub-protocol $P_1$ and $P_2$ respectively. \newline
First we discuss the failure probability of sub-protocol $P_1$. This sub-protocol can fail due to three reasons, similar as in the bipartite setting: the initial states do not belong to the likely subspace or, after the sub-protocol has finished, two or more configurations are compatible with the collected parity information, or the protocol is continued mistakenly after parameter estimation, i.e. the parties should have aborted but continued the multiparty hashing protocol to its very end. \newline
To provide an estimate for the probability that the initial states fall outside the likely subspace w.r.t. sub-protocol $P_1$ we define for color $A$ the random variables $X^{(i)}(b)$ for $1 \leq i \leq N_A$ which take the values 
\begin{align}
X^{(i)}(b) = -\log_2 a^{(b)}_{i} - S(a_{i}) \label{eq:multi:iid:rv}
\end{align}
with probability $a^{(b)}_{i}$. In order to learn $\bm{\mu}$, we observe that a specific $\bm{\mu}=(\mu_1,\ldots,\mu_{N_A})$ belongs to the likely subspace $\mathcal{L}$ whenever each $\mu_i$ belongs to its likely subspace $\mathcal{L}_i$, i.e. 
\begin{align}
\bm{\mu} \in \mathcal{L} \, \Leftrightarrow \, \forall 1 \leq i \leq N_A \colon \, \mu_i \in \mathcal{L}_i .
\end{align}
Consequently 
\begin{align}
\P \left(\bm{\mu} \notin \mathcal{L} \right) \leq \sum^{N_A}_{i=1} \P \left(\mu_i \notin \mathcal{L}_i \right) \leq N_A \max_{1 \leq i \leq N_A} \P \left(\mu_i \notin \mathcal{L}_i \right). \label{eq:multi:outside:0}
\end{align}
We estimate $\P \left(\mu_i \notin \mathcal{L}_i \right)$ via Hoeffding's inequality \cite{Hoeffding}. In order to apply Hoeffding's inequality we need to make sure that $\lambda_{\bm{\mu},\bm{\nu}} \neq 0$ for all $\bm{\mu}$ and $\bm{\nu}$ after twirling, as the the random variables $X^{(i)}(b)$ of (\ref{eq:multi:iid:rv}) need to be bounded. We achieve this by mixing each individual initial state with a small, but defined, portion of the identity operator. From this we observe that the random variables $X^{(i)}$ have zero mean and that $|X^{(i)}| \leq \max_{b \in {0,1}} |\log_2 a^{(b)}_{i}| + S(a_{i}) =: C_i$ after mixing. Therefore the Hoeffding inequality implies 
\begin{align}
\P \left(\left| \sum^n_{k=1} X^{(i)}_k \right| > t \right) \leq 2 \exp\left( \frac{-2t^2}{n C^2_i}\right) \label{eq:multi:outside:1}
\end{align}
for all $t$ where $k$ denotes the index of the initial state within $\rho^{\otimes n}$ and $i$ the $i-$th component of $\bm{\mu}$. Inserting $t = n \delta$ in (\ref{eq:multi:outside:1}) together with $\delta = n^{-1/4}$ yields
\begin{align}
\P \left(\mu_i \notin \mathcal{L}_i \right) = \P \left(\left| \sum^n_{k=1} X^{(i)}_k \right| > n^{3/4} \right) & \leq 2 \exp\left( \frac{-2\sqrt{n}}{C^2_i}\right) \notag \\
& \leq 2 \exp\left( \frac{-2\sqrt{n}}{C^2}\right) \label{eq:multi:outside}
\end{align}
where $C = \max_{1 \leq i \leq N_A} C_i$. Note that (\ref{eq:multi:outside}) is independent of $i$, which implies for (\ref{eq:multi:outside:0}) that
\begin{align}
\P \left(\bm{\mu} \notin \mathcal{L} \right) \leq 2 N_A \exp\left( \frac{-2\sqrt{n}}{C^2}\right).
\end{align}
Observe that $C = \max_{1 \leq i \leq N_A} C_i$ still depends on the initial states. Due to parameter estimation one finds another constant $C' > C$ independent of the initial states. \newline
The probability of not being able to distinguish between two or more configurations is, for a particular component of $\bm{\mu}$, again $2^{-n \delta}$, as for the bipartite case. Hence inserting $\delta = n^{-1/4}$ gives that the probability of misidentifying a specific $\mu_i$ where $1 \leq i \leq N_A$ is bounded by $2^{-n^{3/4}}$. Therefore the probability of misidentifying $\bm{\mu}$ is bounded by $N_A 2^{-n^{3/4}}$. \newline
We point out that also in the multipartite setting a parameter estimation step is crucial in order to ensure distillation. For that purpose we find that the states after twirling and mixing are diagonal within the graph state basis, i.e. of the form 
\begin{align}
\rho = \sum\limits_{\bm{\mu},\bm{\nu}} \lambda_{\bm{\mu},\bm{\nu}} \dm{\psi_{\bm{\mu},\bm{\nu}}}
\end{align}
where all $\lambda_{\bm{\mu},\bm{\nu}} \neq 0$. The goal of parameter estimation is to provide estimates $\overline{a_i}$ and $\overline{b_j}$ for the probability distributions $a_i$ and $b_j$ of (\ref{eq:multi:iid:a}) and (\ref{eq:multi:iid:b}) for all $1 \leq i \leq N_A$ and $1 \leq j \leq N_B$. The concrete boundaries for which the participants continue with hashing depends on the target state of the protocol. However, it suffices to estimate $\lambda_{\bm{\mu},\bm{\nu}}$ for all $\bm{\mu}$ and $\bm{\nu}$ which we denote by $\overline{\lambda_{\bm{\mu},\bm{\nu}}}$. Observe that we have to determine in total $2^N$ coefficients, where $N$ denotes the number of participants and is constant. This can be done via  measurements on $kn$ systems of $\rho$ according to the observables of the correlation operators (\ref{eq:graph:correlation}). Indeed, the expected values of the correlation operators are sufficient to determine the coefficients $\lambda_{\bm{\mu},\bm{\nu}}$ for all $\bm{\mu}$ and $\bm{\nu}$ within $\rho = \sum_{\bm{\mu},\bm{\nu}} \lambda_{\bm{\mu},\bm{\nu}} \dm{\psi_{\bm{\mu},\bm{\nu}}}$. Now one can apply Hoeffding's inequality to exponentially bound the probabilities that the estimates $\overline{\lambda_{\bm{\mu},\bm{\nu}}}$ of $\lambda_{\bm{\mu},\bm{\nu}}$ have a distance larger than some fixed $\eta > 0$ (which corrsponds to the accuracy of our estimate $\overline{\lambda_{\bm{\mu},\bm{\nu}}}$) similar to the bipartite case. From this we deduce that the probability of continuing with the hashing protocol mistakenly is exponentially small in terms of the number $n$ of initial states. \newline
In summary, via the same argument as in the bipartite case (i.e. the previous estimates are upper bounds for the real failure probabilities, see (\ref{eq:conf:local:bi:finalstate}), (\ref{eq:conf:local:ok:1}) and (\ref{eq:conf:local:ok})), the probability that sub-protocol $P_1$ fails satisfies $p_{P_1} \in O(\exp(-\sqrt{n}))$. Similarly one obtains that sub-protocol $P_2$ fails with probability $p_{P_2} \in O(\exp(-\sqrt{n}))$ which implies that $p_f \in O(\exp(-\sqrt{n}))$, thereby proving $\varepsilon_{\mathrm{H}} \in O(\exp(-\sqrt{n}))$ as claimed.
\end{proof}
Observe that Eq. (\ref{eq:conf:multi:local}) is restricted to i.i.d. initial states rather than arbitrary initial states and does not take into account Eve's purification of the initial states. But since Theorem \ref{theo:postselect} of the main text is also applicable to the multiparty hashing protocol, we eliminate these issues and immediately infer for the multiparty hashing protocol prepended by symmetrization by using (\ref{eq:conf:multi:local}) that 
\begin{align}
\|\E^{s} - \F^{s}\|_{\diamond} & \leq 4 (n+kn+1)^{4^M-1} \sqrt{\varepsilon_{\text{H}}} \label{eq:conf:mutli:arb}.
\end{align}
The proof of (\ref{eq:conf:mutli:arb}) is simple: Theorem \ref{theo:postselect} of the main text applies to the multiparty hashing protocol with $d=2^M$, where $M$ denotes the number of participants. Hence (\ref{eq:conf:multi:local}) implies (\ref{eq:conf:mutli:arb}) via Theorem \ref{theo:postselect} of the main text. 

\bibliographystyle{apsrev4-1}
\bibliography{repeater_hashing}

\end{document}